\documentclass[11pt]{article}
\usepackage[margin=1in]{geometry}
\usepackage[T1]{fontenc}
\usepackage{lmodern,microtype}
\usepackage{amsmath,amssymb,amsthm}
\usepackage{booktabs,tabularx}
\usepackage{algorithm,algpseudocode}
\usepackage[authoryear,round]{natbib}
\defcitealias{MS25}{MS25}
\defcitealias{MM26}{MM26}
\usepackage[hidelinks]{hyperref}
\makeatletter
\providecommand{\theHALG@line}{}
\renewcommand{\theHALG@line}{\thealgorithm.\arabic{ALG@line}}
\makeatother
\usepackage[nameinlink,noabbrev,capitalise]{cleveref}
\usepackage{fancyhdr}
\newtheorem{theorem}{Theorem}
\newtheorem{lemma}[theorem]{Lemma}
\newtheorem{proposition}[theorem]{Proposition}
\newtheorem{corollary}[theorem]{Corollary}
\crefname{theorem}{Theorem}{Theorems}
\crefname{appendix}{Appendix}{Appendices}
\Crefname{appendix}{Appendix}{Appendices}
\crefname{lemma}{Lemma}{Lemmas}
\Crefname{lemma}{Lemma}{Lemmas}
\crefname{proposition}{Proposition}{Propositions}
\Crefname{proposition}{Proposition}{Propositions}
\crefname{corollary}{Corollary}{Corollaries}
\Crefname{corollary}{Corollary}{Corollaries}
\crefname{algorithm}{Algorithm}{Algorithms}
\Crefname{algorithm}{Algorithm}{Algorithms}
\algrenewcommand{\algorithmicrequire}{\textbf{Input:}}
\algrenewcommand{\algorithmicensure}{\textbf{Output:}}
\hypersetup{pdftitle={Differentially Private Multicolor Discrepancy and Fair Division of Indivisible Goods},
  pdfauthor={Max Dupr\'{e} la Tour}}
\newcommand{\E}{\mathbb E}
\newcommand{\Prob}{\mathbb P}
\newcommand{\one}{\mathbf 1}
\newcommand{\norm}[1]{\left\lVert#1\right\rVert}
\title{Differentially Private Multicolor Discrepancy\\
and Fair Division of Indivisible Goods}
\author{Max Dupr\'{e} la Tour\\[0.4em]
  \normalsize RIKEN Center for Advanced Intelligence Project\\
  \normalsize The University of Tokyo}
\date{}
\begin{document}
\maketitle

\begin{abstract}
We study the fair division of indivisible goods under pure differential privacy, continuing the line of work initiated by \citet{MS25}. For $n$ agents with nonnegative additive utilities over $m$ goods and a fixed privacy parameter, we give an entry-private algorithm that, with high probability, achieves consensus envy-freeness up to $O(\sqrt n+\log^3 m)$ goods. This substantially improves the dependence on $n$ over the previous $O(n\log m)$ guarantee for ordinary envy-freeness, while providing the stronger consensus guarantee. A key ingredient is a private algorithm for multicolor discrepancy, which may be of independent interest. Our algorithm may require exponential time.

We also obtain substantially stronger guarantees under additional structure: when all item values belong to a public alphabet of size $D$, we give a polynomial-time entry-private algorithm achieving ordinary envy-freeness up to $O(\operatorname{polylog}(mD))$ goods with high probability.

Finally, we prove an $\Omega(\log n)$ lower bound on the number of goods that must be removed to achieve ordinary envy-freeness under entry privacy, for sufficiently many goods, even with binary utilities. 
\end{abstract}

\section{Introduction}
Fair division studies how to allocate resources among agents with
different preferences. In recent years, considerable
attention has been devoted to indivisible goods, for which natural fairness
requirements must be reconciled with the impossibility of splitting
individual items \citep{AABFLMVW23}.

A standard fairness criterion is \emph{envy-freeness}: each agent
values her own bundle at least as much as any other agent's bundle.
An envy-free allocation need not exist when goods are indivisible,
even with two agents and a single good. This motivates
\emph{envy-freeness up to one good} (EF1), which requires that any
envy disappear after removing one good from the envied bundle.
Unlike exact envy-freeness, EF1 is always attainable for monotone
valuations, using the envy-cycle elimination algorithm of
\citet{LMMS04}. More generally, EF$c$ permits the removal of at most
$c$ goods.

These guarantees do not account for the privacy of agents' preferences.
Even if the reported valuations remain confidential, the allocation
itself may disclose information about them. Differential privacy
limits such disclosure by requiring that the distribution of the
published output change only slightly when protected information is
modified \citep{DMNS06}. In fair division, this raises the question
of how much the fairness guarantee must deteriorate to protect
agents' valuations.

\citet{MS25} initiated the study of differentially private fair
division under two notions of privacy. \emph{Entry-level privacy}
protects information about one good in one agent's valuation, whereas
\emph{row-level privacy} protects an agent's entire valuation. Both
notions protect against information revealed by the complete
allocation.

Under entry-level privacy, they gave an $\varepsilon$-differentially
private algorithm achieving EF$c$ with $c=O(n\log m/\varepsilon)$
for $n$ agents with monotone valuations over $m$ goods. Their
algorithm produces connected bundles, meaning intervals in a fixed
ordering of the goods. They also established an
$\Omega(\log m/\varepsilon)$ lower bound for sufficiently many goods.
This lower bound requires connected allocations and does not apply when
arbitrary bundles are allowed.

Under row-level privacy, they proved the lower bound
$c=\Omega\bigl(\min\{m/n,\sqrt{(m/n)\log n}\}\bigr)$,
while their upper bound was the trivial $O(m/n)$ guarantee obtained
by fixing a partition into bundles of nearly equal cardinality
\citep{MS25}.

\paragraph{Consensus fairness and multicolor discrepancy.}
For nonnegative additive valuations, we obtain substantially stronger
entry-private guarantees. Our main result achieves \emph{consensus}
EF$c$: according to every agent, every bundle is worth at least as
much as every other bundle after removing at most $c$ goods from the
latter. Consequently, the allocation is EF$c$ under every assignment
of bundles to agents. Our algorithm achieves
\[
 c=O\!\left(\sqrt n+\frac{\log^3m}{\varepsilon}\right)
\]
with high probability (\cref{thm:real-fairness}). This improves the
dependence on $n$ in the previous entry-private bound while
strengthening the fairness requirement. Moreover, consensus fairness
can require $\Omega(\sqrt n)$ deletions even without privacy
\citep{MM26}. Our bound therefore attains the worst-case nonprivate
benchmark up to an additive polylogarithmic privacy term.

The main technical ingredient is a private multicolor discrepancy
bound, which may be of independent interest. Given a matrix
$V\in[0,1]^{n\times m}$, multicolor discrepancy asks for a partition
of its columns into $k$ bundles whose sums are close in every row.
We give an entry-private algorithm attaining
$\operatorname{Disc}_k(V,A)=O(\sqrt n+\varepsilon^{-1}\log^3m)$
for every $2\le k\le m$ (\cref{thm:main}). For two colors, the 
algorithm adapts the partial-coloring and entropy method of \citet{Spe85}. 
We extend this to $k$ colors using the recursive reduction of \citet{DS03}.
Finally, we derive the consensus-fairness result by adapting the 
discrepancy-to-fairness reduction of \citet{MS22}.

\paragraph{Improved bounds given a public alphabet.}
For ordinary envy-freeness, we obtain a polynomial-time algorithm
under the additional assumption that item values belong to a public
finite alphabet. If the alphabet has size $D\ge2$, the algorithm
achieves entry-private EF$c$ with
\[
 c=O\!\left(
 \frac{\operatorname{polylog}(mD)}{\varepsilon}
 \right).
\]
The dependence on the alphabet is only through its size, with no
restrictions on the magnitudes or spacing of its values. 

\paragraph{Lower bound for binary utilities.}
Ordinary envy-freeness nevertheless incurs an unavoidable cost under
entry-level privacy. We prove an $\Omega(\log n/\varepsilon)$ lower
bound for sufficiently many goods, even for binary valuations and
without any connectivity restriction (\cref{thm:ef-lower}). Thus, at
a fixed privacy parameter and a high constant success probability,
the number of permitted deletions cannot be bounded independently
of $n$ and $m$.
This answers a question posed by \citet{MS25}.

\paragraph{Row-level privacy.}
Finally, we close the gap under row-level privacy by choosing either
a fixed balanced partition or a uniformly random allocation,
depending on the number of goods. This algorithm is
$0$-differentially private and achieves
$c=O\bigl(\min\{m/n,\sqrt{(m/n)\log n}\}\bigr)$, matching the
lower bound while providing the stronger consensus guarantee
(\cref{thm:row-fairness}).

\begin{table}[H]
\centering
\small
\setlength{\tabcolsep}{5pt}
\renewcommand{\arraystretch}{1.2}
\begin{tabularx}{\linewidth}{@{}>{\raggedright\arraybackslash}p{0.19\linewidth}>{\raggedright\arraybackslash}p{0.43\linewidth}>{\raggedright\arraybackslash}X@{}}
\toprule
Fairness & Upper bound & Lower bound \\
\midrule
\multicolumn{3}{@{}l}{\textbf{Entry-level $\varepsilon$-DP}} \\
\addlinespace[5pt]
EF$c$ + Connected
& \mbox{$O(n\log m)$ {\footnotesize(\citetalias{MS25}, Thm.~4.1)}}
& \mbox{$\Omega(\log m)$ {\footnotesize(\citetalias{MS25}, Thm.~4.9)}} \\
\addlinespace[7pt]
EF$c$
& \mbox{$\displaystyle O\!\left(\sqrt n+\log^3m\right)$
  {\footnotesize(\Cref{thm:real-fairness})}}
& \mbox{$\Omega(\log n)$ {\footnotesize(\Cref{thm:ef-lower})}} \\
\addlinespace[7pt]
Consensus EF$c$
& \mbox{$\displaystyle O\!\left(\sqrt n+\log^3m\right)$
  {\footnotesize(\Cref{thm:real-fairness})}}
& \mbox{$\Omega(\sqrt n)$ {\footnotesize(\citetalias{MM26}, Cor.~4.1)}}
  \newline {\footnotesize even without privacy} \\
\addlinespace[7pt]
EF$c$ + Public alphabet of size $D$
& $\displaystyle O\!\left(\operatorname{polylog}(mD)\right)$ {\footnotesize(\Cref{thm:alphabet-ef})}
& \mbox{$\Omega(\log n)$ {\footnotesize(\Cref{thm:ef-lower})}}
  \newline {\footnotesize even for $D=2$} \\
\midrule
\multicolumn{3}{@{}l}{\textbf{Row-level $\varepsilon$-DP}} \\
\addlinespace[5pt]
EF$c$ or\newline Consensus EF$c$
& \mbox{$\displaystyle O\!\left(\sqrt{\frac mn\log n}\right)$
  {\footnotesize(\Cref{thm:row-fairness}; $0$-DP)}}
& \mbox{$\displaystyle \Omega\!\left(\sqrt{\frac mn\log n}\right)$
  {\footnotesize(\citetalias{MS25}, Thm.~3.1)}} \\
\bottomrule
\end{tabularx}
\caption{Deletion bounds for $n\ge2$ agents, $k=n$ bundles, and
$m\ge n^2\log n$ goods.
Fix a small constant $0<\varepsilon\le1$ and assume
$n^{-a}\le\beta<e^{-\varepsilon}/200$ for fixed $a>0$;
upper bounds succeed with probability at least $1-\beta$.
Implicit constants may depend on $\varepsilon$ and $a$.
Utilities are nonnegative and additive, with arbitrary values except
in the public-alphabet row, where the numerical alphabet is fixed in
advance, contains zero, and has arbitrary size $D\ge2$.
The connected upper bound also allows arbitrary monotone utilities.
The lower bounds use binary utilities, rescaled if needed to lie
in the public alphabet.}
\label{tab:results}
\end{table}

\subsection{Technical overview}

\paragraph{Private multicolor discrepancy.}
We first solve the two-color problem, following
Spencer's partial-coloring method \citep{Spe85}. The key is to find
two sign vectors whose row sums are close but which disagree on many
coordinates. Subtracting them cancels the coordinates where they
agree and assigns a color to the others. This gives a \emph{partial
coloring} with small discrepancy. An entropy argument shows that
there are enough suitable pairs to color a constant fraction of the
remaining columns with high probability.

To make this construction private, we first discard pairs that
disagree on too few coordinates, using a test independent of the
matrix. Among the remaining pairs, we favor close row sums. A strict
cutoff on row differences would be problematic: changing one matrix
entry could turn a possible pair into an impossible one. We therefore
give each remaining pair a positive weight, with large row differences
penalized exponentially. The weights change gradually enough under
an entry change to give privacy. We show that sufficiently many pairs
pass the progress test and that large discrepancies remain unlikely,
even after accounting for all rows simultaneously.

These partial colorings let us gradually round fractional assignments
to complete two-way splits. We then split the columns recursively
until there are $k$ bundles, following \citet{DS03}. Early splits
receive smaller privacy budgets because their errors are spread
among more final bundles. Together, the rounding and privacy schedule
preserve the $O(\sqrt n)$ nonprivate term, with an additional
polylogarithmic privacy error proportional to $1/\varepsilon$.

\paragraph{From discrepancy to private consensus fairness.}
A small utility gap may require many deletions if each good is worth
very little. Following \citet{MS22}, we therefore balance two
quantities for each agent: the number of her most valuable goods,
called \emph{head goods}, and the utility of the remaining goods,
measured in units of the smallest head value. We choose enough head
goods that every bundle receives a small reserve of valuable goods.
To remove envy of a bundle, delete its head goods. Its remaining
utility is close to the non-head utility of every other bundle, and
the other bundle's reserve covers any shortfall. Balancing head
counts also bounds the number of deletions.

The obstacle to privacy is the choice of units. Changing one utility
can change the smallest head value, rescaling the entire row of
remaining utilities. Thus a change to one entry in the original
input can change many entries in the matrix we want to balance.

We address this inside the partial-coloring sampler. At each rounding
step, we order the still-fractional non-head goods by value and choose
a noisy cutoff rank near the top of this list. We divide the utilities
below the cutoff by its value, temporarily omitting the goods above
and at the cutoff from this auxiliary row. Only a few goods are
omitted, so their contribution can be covered by a small additional
balancing error.

A single utility change shifts the other goods' ranks only slightly,
even when the cutoff values change greatly. Averaging over nearby
ranks, together with the design of the sampler's weights, keeps the
coloring probabilities stable under this change. Both the rank and
the cutoff value remain hidden. The same rounding and recursive
splitting then balance the two quantities needed for consensus
fairness.

\paragraph{Polynomial-time envy-freeness over a public alphabet.}
We adapt the envy-cycle elimination algorithm of \citet{LMMS04}
using lower estimates of bundle values that never decrease and whose
error is covered by deleting a few goods. The last insertion into each bundle converts
this estimate guarantee into ordinary envy-freeness with one
additional deletion.
When utilities come from a public alphabet, we use classical tools
for differential privacy under continual observation
\citep{DNPR10,CSS11} to maintain these estimates throughout
envy-cycle elimination. For each agent and bundle, we count the goods
that the agent values at or above each alphabet threshold. Deleting a fixed number of the
highest-valued goods reduces each threshold count by that number,
down to a minimum of zero. Thus a single deletion set simultaneously
absorbs the error at every threshold. This makes the fairness bound
independent of the magnitudes and spacing of the alphabet values.

\paragraph{The entry-private lower bound.}
With zero-one utilities, envy is a difference between two counts:
the number of goods an agent values in someone else's bundle and
the number she values in her own. Each deletion reduces this gap
by at most one. We show that privacy makes a large gap unavoidable
for at least one agent.

We generate utilities by independent fair coin flips on a chosen set
of goods, giving all other goods value zero. Fix any allocation the
mechanism might return. Observing it can make some coin outcomes
more likely, since the mechanism used the utilities to choose the
bundles. Privacy limits how far any one coin can be biased, even
when all the other outcomes are known.

For the analysis, we introduce a new coin-flip experiment with the
same fixed bundles. Its coins are independent and biased toward
reducing envy: an agent is more likely to value her own goods and
less likely to value other agents' goods. We choose the most
favorable biases allowed by the privacy bound. This can only make
large envy gaps less likely, so a lower bound for this experiment
also applies to the actual mechanism.

Now compare the half of the agents receiving the fewest randomly
valued goods with a bundle containing the most such goods. Even
with the favorable biases, random fluctuations can make an agent
value that larger bundle substantially more than her own. The
fluctuations are independent across agents in our experiment. A
larger gap is rarer, but having many agents gives it more chances
to occur. Choosing the number of random goods appropriately,
a binomial tail bound shows that some gap reaches order
$\log(n/\beta)/\varepsilon$ with probability greater than $\beta$.
Each deletion removes at most one unit of this gap, which gives
the lower bound.

\paragraph{Row-level privacy.}
Here the allocation ignores utilities. Under independent uniform
assignment, an ordinary concentration bound controls each of the
two bounded rows per agent in the discrepancy-to-fairness reduction.
That reduction then bounds the deletions needed to remove envy.
Choosing between this random allocation and a fixed balanced
partition yields the tight row-private bound with privacy parameter zero.

\subsection{Related work}\label{sec:related-work}
Privacy in allocation has also been studied under joint differential
privacy. \citet{HHRRW14} gave algorithms for approximately
welfare-maximizing allocations for gross-substitutes valuations in
markets with sufficiently many copies of each good. Joint privacy
protects an agent's valuation
against the combined outputs received by all other agents, while
allowing her own allocation to depend on that valuation.

The fair-division models of \citet{MS25}, which we adopt,
protect information revealed by the complete allocation and measure
fairness through the number of goods whose removal eliminates envy.

Our discrepancy construction builds on the partial-coloring method
underlying Spencer's theorem \citep{Spe85}. \citet{DS03} developed
a recursive passage from two-color to multicolor discrepancy.
Polynomial-time algorithms attaining Spencer's bound were
subsequently obtained by \citet{Ban10} and \citet{LM12}.

The connection to fair division was developed by
\citet{MS22}, who used multicolor discrepancy to obtain consensus
divisions and almost envy-free allocations among groups. Their
reduction is the starting point for our private consensus algorithm.
More recently, \citet{MM26} proved a tight $\Omega(\sqrt n)$ worst-case
lower bound for multicolor discrepancy for every number of colors $k\ge2$,
yielding the same lower bound for nonprivate consensus fairness.

\section{Preliminaries}\label{sec:preliminaries}
Write $[r]=\{1,\ldots,r\}$. There are $m\ge2$ goods and $n$ agents;
$k$ denotes the number of bundles. An allocation
$A=(A_1,\ldots,A_k)$ is an ordered partition of $[m]$, with empty
bundles allowed. The letter $C$ denotes a universal constant that may
increase from line to line.

For a real number $x$, write $(x)_+=\max\{x,0\}$ and let
$\operatorname{sgn}(x)$ be $-1$, $0$, or $1$ according to its sign.
We use $\one_{\mathcal E}$ for the indicator of an event $\mathcal E$,
$\one_A$ for the indicator vector of a set $A$, and $\one$ for the
all-ones vector, with the ambient coordinates clear from context.
For a matrix $B$, the notation $B_T$ restricts it to the columns in
$T$, and $B_{i,\cdot}$ denotes its $i$th row.
The matrix $\operatorname{diag}(d)$ has diagonal entries $d_g$ and
zero off-diagonal entries.

\subsection{Fairness}
Each agent $i$ has a nonnegative additive valuation function $v_i$,
so $v_i(S)=\sum_{g\in S}v_i(g)$. We view a matrix
$V\in\mathbb R_{\ge0}^{n\times m}$ as defining the valuations
$v_1,\ldots,v_n$ by $v_i(g)=V_{ig}$.
Let $c$ be a nonnegative integer. When $k=n$, agent $i$ receives
$A_i$. The allocation is
\emph{envy-free up to $c$ goods} (EF$c$) if
\[
 \forall i,j\quad \exists R\subseteq A_j,\quad |R|\le c,\quad
 v_i(A_i)\ge v_i(A_j\setminus R).
\]
A \emph{consensus $1/k$-division up to $c$ goods}, or
\emph{consensus EF$c$}, satisfies every agent's comparison between
every ordered pair of bundles:
\[
 \forall i,j,\ell\quad \exists R\subseteq A_\ell,\quad |R|\le c,\quad
 v_i(A_j)\ge v_i(A_\ell\setminus R).
\]
The deletion set may depend on the agent and the compared bundles;
it is not released. When $k=n$, consensus EF$c$ implies ordinary
EF$c$ under every assignment of agents to bundles \citep{MS22}.

\subsection{Discrepancy}
For discrepancy, we assume $V\in[0,1]^{n\times m}$ and measure
the largest difference between two bundle values:
\[
 \operatorname{Disc}_k(V,A)=\max_{i\in[n]}\max_{j,\ell\in[k]}
          \bigl(v_i(A_j)-v_i(A_\ell)\bigr).
\]
Discrepancy is sometimes defined by deviation from an equal share.
The two definitions differ by at most a factor of two:
\[
 \frac12\operatorname{Disc}_k(V,A)\le
 \max_{i,j}\left|v_i(A_j)-\frac{v_i([m])}{k}\right|
 \le \operatorname{Disc}_k(V,A).
\]
A \emph{signing} is a vector $\sigma\in\{-1,1\}^m$. It induces
a two-color partition $A$ by placing good $g$ in $A_1$ when
$\sigma_g=1$ and in $A_2$ when $\sigma_g=-1$. This partition has
discrepancy
\[
 \operatorname{Disc}_2(V,A)=\|V\sigma\|_\infty.
\]
A \emph{partial coloring}
$z\in\{-1,0,1\}^m$ leaves the zero coordinates unassigned.
In \emph{fractional rounding}, a vector $x\in[0,1]^m$ is
replaced by $y\in\{0,1\}^m$, with error $\|V(y-x)\|_\infty$.

\subsection{Differential privacy}\label{sec:privacy-tools}
Two input matrices are \emph{entry-adjacent} if they differ in one
entry, and \emph{row-adjacent} if they differ in one entire row.
The change may be arbitrary within the input domain: $[0,1]$ for
discrepancy, and the nonnegative reals for the general fairness results.
In \cref{sec:alphabet}, utilities instead belong to a fixed public
finite alphabet, and replacements must stay within that alphabet.
For either adjacency relation, a mechanism $\mathcal M$ is
$\varepsilon$-differentially private (DP) if
\[
 \Prob[\mathcal M(X)\in\mathcal E]\le e^\varepsilon
       \Prob[\mathcal M(X')\in\mathcal E]
\]
for every adjacent pair $X,X'$ and every output event $\mathcal E$.
All guarantees in this paper are pure DP; $\beta$ is the failure
probability of the accuracy guarantee. Subscripted versions, such as
$\varepsilon_0,\varepsilon_t$ and $\beta_0$, denote the corresponding
budgets for individual components or calls.
Our allocation mechanisms release only the allocation; internal
samples and running times are outside the output model.

We use the following standard privacy tools \citep{DR14}.

\paragraph{Postprocessing.}
Processing a private output without further access to the input
preserves its privacy guarantee. The processing may use public
information and fresh randomness.

\paragraph{Composition.}
Privacy costs add across successive calls, even when later calls
depend on earlier outputs. More precisely, suppose that after fixing
all previous outputs, call $t$ is $\varepsilon_t$-DP. If
$\sum_t\varepsilon_t\le\varepsilon$ on every transcript, releasing
the complete transcript is $\varepsilon$-DP
\citep[Theorem~B.1]{DR14}.

Under entry adjacency, independently sampled mechanisms on fixed
disjoint sets of entries compose with the maximum privacy cost:
only one mechanism sees a changed entry.
In a recursion, we compare the same transcript of internal outputs on
neighboring inputs. This fixes the sets used by later calls. A call
on unchanged data has the same conditional law on both inputs, so
only calls on changed data contribute to privacy loss.

\paragraph{The exponential mechanism.}
We use the following weighted form of the exponential mechanism
\citep[Theorem~3.10]{DR14}.

\begin{lemma}[Exponential mechanism]\label[lemma]{lem:exponential-mechanism}
Let $\Omega$ be a fixed finite output set, let $Q$ be a
data-independent probability distribution on it, and let $a\ge0$.
Suppose positive weights satisfy
\[
 |\log w_X(o)-\log w_{X'}(o)|\le a
\]
for all adjacent inputs $X,X'$ and all $o$ in the support of $Q$.
Then the mechanism with output law
\[
 \Prob[\mathcal M(X)=o]
 =\frac{Q(o)w_X(o)}{\sum_{o'\in\Omega}Q(o')w_X(o')}
\]
is $2a$-DP.
\end{lemma}
\begin{proof}
For adjacent inputs, the ratio of the two weights lies in
$[e^{-a},e^a]$ at every output in the support of $Q$. Summing
against $Q$ gives the same bounds for the ratio of the normalizing
constants. Each output probability therefore changes by a factor
of at most $e^{2a}$. Summing over any output event proves the claim.
\end{proof}
We use this lemma for privacy and analyze accuracy separately.

\paragraph{The Laplace mechanism.}
A vector query with $\ell_1$ sensitivity at most $s$ can be released
with $\varepsilon$-DP by adding independent centered Laplace noise
of scale $s/\varepsilon$, denoted by $\operatorname{Lap}(s/\varepsilon)$,
to each coordinate \citep[Theorem~3.6]{DR14}.

\section{Private multicolor discrepancy}\label{sec:model}

Fix $V\in[0,1]^{n\times m}$. Throughout this section and
\cref{sec:fairness}, the parameters $n\ge1$, $2\le k<m$,
$0<\varepsilon\le1$, and $0<\beta<1/2$ are public.
The case $k=1$ is trivial. For $k\ge m$, a fixed partition into
singleton and empty bundles gives discrepancy at most one
and consensus EF1.
Universal constants in the algorithms are fixed before seeing the input;
their required sizes are determined by the proofs below.
Restricted matrices and their sign vectors retain the original column labels.

The main result of this section is the following private multicolor discrepancy bound.

\begin{theorem}[Private multicolor discrepancy]\label{thm:main}
There is an entry-level $\varepsilon$-DP algorithm whose output satisfies
\[
 \operatorname{Disc}_k(V,A)\le C\left(
 \sqrt n+\frac{\log^2 m}{\varepsilon}
                  \log\frac{nkm}{\beta}\right)
\]
with probability at least $1-\beta$. Its expected running time is
$m^{O(n)}\operatorname{poly}(n,k)$ in the computation model described
in \cref{sec:sampling}.
\end{theorem}

The leading $\sqrt n$ term matches the optimal nonprivate bound
that depends only on the number of rows \citep{DS03,MM26}.

We build the algorithm in three steps. First, we construct a private
partial coloring: it assigns signs to a constant fraction of the
columns while keeping each row sum small. This follows Spencer's
partial-coloring framework \citep{Spe85}. We present the nonprivate
sampler first, then modify its weights to obtain privacy and stability
under row contraction. These guarantees will also support the
fairness sampler in \cref{sec:fairness}. Second,
a rounding procedure combines partial colorings into
an integral vector. Finally, we adapt the recursive reduction of
\citet{DS03} to obtain a multicolor partition.

\subsection{Sampling partial colorings}

\subsubsection{A nonprivate partial-coloring sampler for two colors}\label{sec:collision-bound}
Let $B\in[0,1]^{n\times q}$, $q\ge1$, contain the columns still to
be colored. Draw independent uniform centered sign vectors
$\sigma,\tau\in\{-1/2,1/2\}^q$. Their difference
\[
 \sigma-\tau\in\{-1,0,1\}^q
\]
colors a coordinate exactly when the two signs differ there.
Write $Q$ for this uniform pair distribution. Under $Q$, the number
of colored coordinates has distribution $\operatorname{Bin}(q,1/2)$, so
\[
 \Prob_Q[|\{g:\sigma_g\ne\tau_g\}|<q/4]\le e^{-q/16}.
\]
We will accept pairs with close row sums. To allow an error of order
$w>0$, first collapse the interval $[-w,w]$ to zero and then compare
the transformed sums using a triangle weight:
\[
 \phi_w(x)=\operatorname{sgn}(x)(|x|-w)_+,
 \qquad T_w(x)=(1-|x|/w)_+.
\]
The map $\phi_w$ is odd and $1$-Lipschitz, and
$|x-\phi_w(x)|\le w$.

One trial draws $(\sigma,\tau)\sim Q$ and accepts with probability
\[
 \prod_{i=1}^n
 T_w\bigl(\phi_w((B\sigma)_i)-\phi_w((B\tau)_i)\bigr).
\]
Repeat until acceptance, then return $\sigma-\tau$. A positive
acceptance weight means that, in every row, the two transformed sums
differ by less than $w$. Thus every output satisfies
$\norm{B(\sigma-\tau)}_\infty\le3w$. The remaining question is whether acceptance
still leaves enough pairs that color many coordinates.

We choose the width to retain enough accepted pairs for the progress
bound, while ensuring that the widths sum to $O(\sqrt n)$ as the
number of remaining columns decreases geometrically.
Choose a sufficiently large universal constant $C_0$ and set
\begin{equation}\label{eq:kernel-width}
 w=\begin{cases}
 C_0n/\sqrt q,&q\ge n,\\
 C_0(nq)^{1/4},&q<n.
 \end{cases}
\end{equation}

\begin{lemma}[Nonprivate partial coloring]\label[lemma]{lem:nonprivate-partial}
With the width in \eqref{eq:kernel-width}, one trial of the nonprivate
sampler is accepted with probability at least $e^{-q/32}$.
The sampler returns a partial coloring with discrepancy at most $3w$,
and the probability that it colors fewer than $q/4$ coordinates is
at most $e^{-q/32}$.
\end{lemma}
\begin{proof}
We bound the acceptance probability by encoding row sums with a small
amount of information. This is the entropy step in the partial-coloring
method \citep[Section~13.2]{AS08}, applied to the sampler just defined.

For row $i$ and a shift $\theta\in[0,w)$, record
\[
 I_{i,\theta}(\sigma)
 =\left\lfloor\frac{\phi_w((B\sigma)_i)+\theta}{w}\right\rfloor.
\]
For fixed transformed sums $Y,Y'$, a uniform shift puts them in the
same interval with probability $T_w(Y-Y')$. Indeed, if
$|Y-Y'|<w$, the shifts that place a grid boundary between them have
total length $|Y-Y'|$ in the interval $[0,w)$; all other shifts
put them in the same cell. Their agreement probability is therefore
$1-|Y-Y'|/w$. If $|Y-Y'|\ge w$, they cannot share a cell, so the
probability is zero. Together these cases give
$(1-|Y-Y'|/w)_+=T_w(Y-Y')$. Therefore the sampler's
acceptance probability is exactly the probability that all row
records agree, averaged over independent uniform shifts.

Let $\operatorname{Ent}$ denote the entropy
$-\sum_j p_j\log p_j$ of a discrete distribution. Fix a shift.
For $j\ge1$, a record at least $j$ requires
$(B\sigma)_i\ge jw$, and a record at most $-j$ requires
$(B\sigma)_i\le-jw$.
Hoeffding's inequality gives
\[
 \max\bigl\{\Prob[I_{i,\theta}\ge j],\
             \Prob[I_{i,\theta}\le-j]\bigr\}\le e^{-2j^2w^2/q}.
\]
The indicators of these events determine the record, since
\[
 I_{i,\theta}
 =\sum_{j\ge1}\one_{\{I_{i,\theta}\ge j\}}
  -\sum_{j\ge1}\one_{\{I_{i,\theta}\le-j\}}.
\]
A binary
variable with success probability $p$ has entropy at most $2\sqrt p$.
Entropy subadditivity consequently gives
\[
 \operatorname{Ent}(I_{i,\theta})
 \le4\sum_{j\ge1}e^{-j^2w^2/q}
 \le\frac{C\sqrt q}{w}e^{-w^2/(2q)}.
\]
For the last inequality, $j^2\ge(1+j^2)/2$ gives
\[
 \sum_{j\ge1}e^{-j^2w^2/q}
 \le e^{-w^2/(2q)}\int_0^\infty e^{-x^2w^2/(2q)}\,dx
 =\frac{\sqrt{\pi q/2}}{w}e^{-w^2/(2q)}.
\]
The integral bounds the remaining sum because its integrand is
decreasing. A second application of subadditivity bounds the
entropy of all row records by
\[
 \frac{Cn\sqrt q}{w}e^{-w^2/(2q)}.
\]
For $q\ge n$, substituting $w=C_0n/\sqrt q$ and dropping the
exponential gives $Cq/C_0$. For $q<n$, the function
$x^{3/2}e^{-x/2}$ is bounded on $[0,\infty)$; applying this with
$x=w^2/q$ gives $e^{-w^2/(2q)}\le Cq^{3/2}/w^3$ and hence
\[
 \frac{Cn\sqrt q}{w}e^{-w^2/(2q)}
 \le\frac{Cnq^2}{w^4}=\frac{Cq}{C_0^4}.
\]
Both bounds are at most $q/32$ for sufficiently large $C_0$.

If the joint record has probabilities $(p_j)$, two independent
records agree with probability
\[
 \sum_jp_j^2\ge
 \exp\!\left(\sum_jp_j\log p_j\right)
 \ge e^{-q/32},
\]
by Jensen's inequality. This holds for every choice of shifts, so
averaging proves the acceptance bound. Conditioning on acceptance
can increase the probability of coloring fewer than $q/4$ coordinates
by at most $e^{q/32}$. Its unconditional probability is at most
$e^{-q/16}$, giving the claimed $e^{-q/32}$ bound. The discrepancy
bound follows directly from the acceptance rule.
\end{proof}

\subsubsection{Sampling a private partial coloring}\label{sec:partial-colorings}
We now modify this same sampler. The triangle weight can vanish
after a small input change, so it cannot directly give pure privacy.
We smooth it by averaging triangles of random widths. For
$\lambda>0$, let $L=w+\operatorname{Exp}(\lambda)$, where the
exponential variable has rate $\lambda$, and define
\[
 K_{\lambda,w}(x)
 =\frac{\E[(L-|x|)_+]}{\E L}
 =\frac{\E[L\,T_L(x)]}{\E L}.
\]
Thus $K_{\lambda,w}$ averages triangle weights under the law of $L$
reweighted by the factor $L/\E L$.
Since $\E L=w+1/\lambda$, it remains to compute the numerator.
If $|x|\le w$, then $L\ge w\ge|x|$, so
\[
 \E[(L-|x|)_+]=\E L-|x|=w+\frac1\lambda-|x|.
\]
If $|x|>w$, the tail-integral formula for a nonnegative random
variable gives
\[
 \begin{aligned}
 \E[(L-|x|)_+]
 &=\int_0^\infty\Prob[L>|x|+t]\,dt\\
 &=\int_0^\infty e^{-\lambda(|x|-w+t)}\,dt
 =\frac{e^{-\lambda(|x|-w)}}{\lambda}.
 \end{aligned}
\]
Dividing these two expressions by $w+1/\lambda$ gives
\[
 K_{\lambda,w}(x)=
 \begin{cases}
 1-\dfrac{\lambda|x|}{1+\lambda w},&|x|\le w,\\[2mm]
 \dfrac{e^{-\lambda(|x|-w)}}{1+\lambda w},&|x|>w
 \end{cases}.
\]
The width $w$ still controls discrepancy, while $\lambda$ controls
privacy and tail decay. The weight is even, decreases with $|x|$,
and satisfies $T_w\le K_{\lambda,w}\le1$. On $(0,w)$ its absolute
logarithmic derivative is $\lambda/[1+\lambda(w-x)]\le\lambda$;
beyond $w$ it is $\lambda$. Hence
\[
 |\log K_{\lambda,w}(x)-\log K_{\lambda,w}(y)|
 \le\lambda|x-y|.
\]
For fixed $\lambda,w>0$ and a row $b\in[0,1]^q$, the corresponding
weight on a pair of sign vectors is
\begin{equation}\label{eq:smooth-penalty}
 K_{\lambda,w}\bigl(\phi_w(b\cdot\sigma)-\phi_w(b\cdot\tau)\bigr).
\end{equation}
The factor equals one when both row sums lie in $[-w,w]$.
Since $\phi_w$ is $1$-Lipschitz, changing either row sum changes
the logarithm of this factor by at most $\lambda$ times that change.
Before imposing the progress requirement, the weighted pair law
assigns probability
\begin{equation}\label{eq:gibbs}
 \frac{Q(\sigma,\tau)\prod_{i=1}^n
 K_{\lambda,w}\bigl(\phi_w((B\sigma)_i)-\phi_w((B\tau)_i)\bigr)}
 {Z_B(\lambda,w)}
\end{equation}
to each pair $(\sigma,\tau)$, where $Z_B(\lambda,w)$ is the
expectation of the product weight under $Q$.
Write $\mathcal P$ for the event that $\sigma$ and $\tau$ differ
on at least $q/4$ coordinates. We reject every proposal outside
$\mathcal P$, so \cref{alg:partial-coloring} samples
\eqref{eq:gibbs} conditioned on $\mathcal P$. This test depends only
on the proposed pair and the public number of columns.

\begin{algorithm}[H]
\caption{\textsc{PartialColor}: private partial coloring}
\label[algorithm]{alg:partial-coloring}
\small
\begin{algorithmic}[1]
\Require $B\in[0,1]^{n\times q}$, $q\ge1$, and public $\lambda,w>0$
\Ensure A $2\lambda$-DP partial coloring with at least $q/4$ nonzero coordinates.
\Function{PartialColor}{$B,\lambda,w$}
  \Repeat
    \State Draw independent uniform $\sigma,\tau\in\{-1/2,1/2\}^q$
    \State Reject and restart the trial if $|\{g:\sigma_g\ne\tau_g\}|<q/4$.
    \State Accept the pair with probability
      $\prod_{i=1}^n K_{\lambda,w}
      \bigl(\phi_w((B\sigma)_i)-\phi_w((B\tau)_i)\bigr)$
  \Until{the pair is accepted}
  \State \Return $\sigma-\tau$
\EndFunction
\end{algorithmic}
\end{algorithm}

To turn a small row weight on an error event into a small probability,
we must also control normalization. We compare the product of all
other row factors with the product after restoring that row, and
bound how much the normalizing constant can decrease. This also
controls the rejection sampler's running time in \cref{sec:sampling}.
We first treat products of row factors, then extend the bounds to
the averages used by the fairness sampler.

\begin{lemma}[The cost of adding one row]\label[lemma]{lem:row-factor}
Fix $q\ge1$ and $\lambda,w>0$. For each row $b\in[0,1]^q$, write
\[
 F_b(\sigma,\tau)
 =K_{\lambda,w}\bigl(\phi_w(b\cdot\sigma)-\phi_w(b\cdot\tau)\bigr)
\]
for the weight in \eqref{eq:smooth-penalty}.

\emph{(i) Products of row factors.}
Let $W$ be a finite product of row factors, with the empty product
equal to one, and define $\nu=QW/\E_QW$.
Adding any row $b\in[0,1]^q$, meaning replacing $W$ by $WF_b$,
changes the normalizing constant by the ratio
\[
 \frac{\E_Q[WF_b]}{\E_QW}
 =\E_\nu F_b\ge\frac1{1+q\lambda}.
\]
Consequently, if $F_b\le a$ on an event $\mathcal E$, for $a\ge0$,
then the law proportional to $QWF_b$ assigns this event probability
at most $(1+q\lambda)a$.

\emph{(ii) Probability mixtures.}
Both conclusions also hold when $W$ is formed from the constant one
and row factors by repeated finite products and finite or countable
probability mixtures, and the added factor $F_b$ is replaced by any
probability mixture $F$ of row factors.
A probability mixture is a pointwise average
$\sum_jp_jW_j(\sigma,\tau)$, where $p_j\ge0$, $\sum_jp_j=1$, and
the coefficients do not depend on $(\sigma,\tau)$.
Mixtures are taken before normalization: with the same definition
$\nu=QW/\E_QW$, if $W=\sum_jp_jW_j$, then $\nu$ mixes the laws
$QW_j/\E_QW_j$ with probabilities $p_j\E_QW_j/\E_QW$.
\end{lemma}
\begin{proof}
\emph{(i) Products of row factors.}
We represent the earlier weight $W$ by a random partition of the
sign vectors. A \emph{collision} means that two independent uniform
sign vectors belong to the same cell of this partition.
Adding a row refines its cells, meaning that it divides each old
cell into smaller cells using a fresh grid.
Dividing a cell into at most $N$ pieces reduces its contribution to
the collision probability by at most a factor $N$.

First consider a single row factor. Use the representation of
$K_{\lambda,w}$ as an average of triangles: draw a spacing from
the law of $L$ reweighted by $L/\E L$, then shift a grid of that
spacing uniformly. For a fixed grid, group sign vectors according
to the interval containing their transformed row sum. This gives a
partition of the sign vectors. As in \cref{sec:collision-bound},
averaging over the grid shows that the row factor is exactly the
probability that the two signs belong to the same cell.

For the product $W$, draw the component partitions independently and
intersect their cells: the signs lie in the same resulting cell
exactly when they do so in every component partition. The empty
product is represented by the partition with a single
cell containing all signs. Thus $W$ itself is the probability that
the two signs share a cell of a random partition.

For the additional row $b$, draw a fresh independent grid by the
same construction. Since $b\in[0,1]^q$, its centered
projections lie in $[-q/2,q/2]$. The map $\phi_w$ is $1$-Lipschitz,
so all transformed projections lie in a fixed interval of length
at most $q$. Let $N$ be the number of fresh grid intervals meeting
this interval, and let $D_1,\ldots,D_N$ be the corresponding cells
of sign vectors. This choice of $N$ depends only on the fresh grid,
not on the earlier partition.

Fix both partitions. For every old cell $\mathcal C$,
Cauchy--Schwarz gives
\[
 \sum_{j=1}^N|\mathcal C\cap D_j|^2
 \ge\frac{|\mathcal C|^2}{N}.
\]
Summing over old cells and dividing by $4^q$ compares the probability
under $Q$ of sharing a cell in both partitions with the probability
of sharing an old cell. Since the old partition and the fresh grid
are independent, averaging yields
\[
 \E_Q[WF_b]\ge\E\!\left[\frac1N\right]\E_QW,
 \qquad\text{hence}\qquad
 \E_\nu F_b\ge\E\!\left[\frac1N\right].
\]

It remains to bound the expected number of intervals. Apart from
shifts placing a boundary exactly at an endpoint, which have
probability zero, the number of cells meeting an interval is one
plus the number of grid boundaries inside it. For fixed spacing $L$,
each subinterval of length $L$ contains exactly one boundary.
Any remaining subinterval of length less than $L$ contains a boundary
with probability equal to its length divided by $L$, because the shift
is uniform. Thus the expected boundary count is the interval's length
divided by $L$, which is at most $q/L$. The spacing law
is weighted by $L/\E L$, so its expected reciprocal is
$\E[L(1/L)]/\E L=1/\E L$. Consequently,
\[
 \E N\le\frac{\E[L(1+q/L)]}{\E L}
       =1+\frac q{\E L}
       =1+\frac q{w+1/\lambda}\le1+q\lambda.
\]
Applying Jensen's inequality to the collision bound yields
\[
 \E_\nu F_b
 \ge\E\frac1N\ge\frac1{\E N}\ge\frac1{1+q\lambda}.
\]
For the event bound, reweighting $\nu$ by $F_b$ gives
\[
 \frac{\E_\nu[F_b\one_{\mathcal E}]}{\E_\nu F_b}
 \le(1+q\lambda)a.
\]

\emph{(ii) Probability mixtures.}
The random-partition representation also survives taking a
probability mixture: first choose a component with its specified
probability, then draw that component's partition. This works for
finite and countable mixtures. Together with the product
construction from part (i), it represents every weight $W$ allowed
in this part. Adding a single row still means drawing a fresh
independent grid, so the same collision argument gives
$\E_\nu F_b\ge(1+q\lambda)^{-1}$.

For a mixture $F=\sum_jp_jF_{b_j}$ of added row factors, averaging
this bound gives
\[
 \E_\nu F=\sum_jp_j\E_\nu F_{b_j}
 \ge\frac{\sum_jp_j}{1+q\lambda}
 =\frac1{1+q\lambda}.
\]
Finally, if $F\le a$ on $\mathcal E$, then
\[
 \frac{\E_\nu[F\one_{\mathcal E}]}{\E_\nu F}
 \le(1+q\lambda)a.\qedhere
\]
\end{proof}

We now collect the sampler's guarantees. The contraction bound
also allows a perturbation of the rescaled row, as needed for
nearby cutoff ranks in \cref{sec:fairness}.

\begin{lemma}[Private partial coloring and contraction stability]\label[lemma]{lem:gibbs}
For every $\lambda,w>0$, \cref{alg:partial-coloring} is
$2\lambda$-DP and always colors at least $q/4$ coordinates.
With the width in \eqref{eq:kernel-width}, its accepted pair satisfies
\[
 \Prob[|(B(\sigma-\tau))_i|>3w+t]
 \le33(1+q\lambda)e^{-\lambda t}
\]
for every row $i$ and $t\ge0$. For this width, a pair drawn from
\eqref{eq:gibbs} lies outside $\mathcal P$ with probability at most
$e^{-q/32}$.

The unnormalized row factors $F_b$ from \cref{lem:row-factor} satisfy,
for every $b,c\in[0,1]^q$, scalar $\rho\in[0,1]$, and fixed
$\sigma,\tau\in\{-1/2,1/2\}^q$,
\begin{equation}\label{eq:contraction}
 F_c(\sigma,\tau)
 \ge e^{-\lambda\norm{c-\rho b}_1}F_b(\sigma,\tau).
\end{equation}
Thus contracting the entire row by the common scalar $\rho$ cannot
decrease its weight. Taking $\rho=1$ and interchanging $b,c$ gives
\[
 e^{-\lambda\norm{c-b}_1}\le F_c/F_b
 \le e^{\lambda\norm{c-b}_1}.
\]
\end{lemma}

\begin{proof}[Proof of \cref{lem:gibbs}]
\emph{Contraction stability.}
For $0\le\rho\le1$, the derivative of
$t\mapsto\phi_w(\rho t)$ is
$\rho\one_{\{|\rho t|>w\}}$ almost everywhere, bounded by
the derivative $\one_{\{|t|>w\}}$ of $\phi_w(t)$.
Integrating between $x$ and $y$ gives
\[
 |\phi_w(\rho x)-\phi_w(\rho y)|
 \le |\phi_w(x)-\phi_w(y)|.
\]
Since $K_{\lambda,w}$ decreases with the absolute value of its
argument, $F_{\rho b}\ge F_b$.
Each centered projection changes by at most
$\norm{c-\rho b}_1/2$ when replacing $\rho b$ by $c$.
The logarithmic Lipschitz bound therefore gives
\[
 F_c\ge e^{-\lambda\norm{c-\rho b}_1}F_{\rho b}
      \ge e^{-\lambda\norm{c-\rho b}_1}F_b.
\]
This proves \eqref{eq:contraction}. Taking $\rho=1$ and
interchanging $b,c$ gives the two-sided perturbation bound.

\emph{Privacy.}
An entry change has row $\ell_1$ distance at most one, so the
stability bound shows that the logarithm of the product weight
changes by at most $\lambda$. The progress test retains the same set $\mathcal P$ on
both neighboring inputs. Apply \cref{lem:exponential-mechanism}
with the uniform law on this set as its public base distribution
and $a=\lambda$: the restricted weights and their normalizing
constants each change by at most a factor $e^\lambda$, giving
$2\lambda$-DP. This proves privacy directly for the restricted law.

\emph{Conditioning on progress.}
First consider the law \eqref{eq:gibbs}, without the progress check.
Since $K_{\lambda,w}\ge T_w$, \cref{lem:nonprivate-partial} gives
$Z_B\ge e^{-q/32}$ for the prescribed width. Product weights are
at most one, so the probability of fewer than $q/4$ disagreements
under this law is at most $e^{-q/16}/Z_B\le e^{-q/32}$.
Thus, with probabilities in the next two displays taken under
\eqref{eq:gibbs},
\[
 \Prob[\mathcal P]\ge1-e^{-q/32}
 \ge\frac{q}{32+q}\ge\frac1{33}.
\]
Here $e^x\ge1+x$ gives the middle inequality, and $q\ge1$ the last.

Rejecting pairs outside $\mathcal P$ removes exactly those outcomes
and renormalizes the probabilities of the remaining ones. Therefore,
for any event $\mathcal E$,
\[
 \Prob[\mathcal E\mid\mathcal P]
 =\frac{\Prob[\mathcal E\cap\mathcal P]}{\Prob[\mathcal P]}
 \le33\Prob[\mathcal E].
\]
We can consequently analyze row errors under the original weighted
law and multiply the resulting bound by $33$. The retained set is
nonempty and all kernel weights are positive, so rejection terminates
almost surely and every returned coloring meets the progress requirement.

\emph{Row tails.}
Fix row $i$. If $|(B(\sigma-\tau))_i|>3w+t$, the transformed row sums
differ by more than $w+t$, so this row's weight is at most
\[
 K_{\lambda,w}(w+t)
 =\frac{e^{-\lambda t}}{1+\lambda w}\le e^{-\lambda t}.
\]
Apply the event bound in part~(i) of \cref{lem:row-factor} with
$W$ the product of the other row factors. Under the law
\eqref{eq:gibbs}, the failure probability is at most
$(1+q\lambda)e^{-\lambda t}$. Conditioning on $\mathcal P$ multiplies
this bound by at most $33$, proving the stated output guarantee.
\end{proof}

\subsection{Rounding fractional vectors}\label{sec:rounding}

Rounding fractional vectors using discrepancy bounds is standard
\citep{DS03}. We give the rounding argument separately from the
private sampler, so that it can also be used in \cref{sec:fairness}.
Starting from $(1/2,\ldots,1/2)$ gives a two-color signing; other
starting vectors give the unequal splits needed for multicolor discrepancy.

Fix $B\in[0,1]^{n\times q}$ and a public initial vector
$x\in[0,1]^q$. Let $y\in[0,1]^q$ be the current vector and
$T=\{g:0<y_g<1\}$ its fractional coordinates. Already integral
coordinates will stay fixed. For each $g\in T$, put
\[
 d_g=\min\{y_g,1-y_g\}.
\]
This is the distance to the nearer endpoint of $[0,1]$.
For any partial coloring $z\in\{-1,0,1\}^T$, both
$y+d\odot z$ and $y-d\odot z$ lie in $[0,1]^q$;
the updates are zero outside $T$, and $\odot$ denotes coordinatewise
multiplication. Each nonzero coordinate of $z$ reaches an endpoint
in at least one of these two directions. Choosing the direction that
fixes more coordinates therefore makes at least half the nonzero
coordinates integral.

Both directions have the same absolute row error,
$|B_T(d\odot z)|$. To control this error, we apply the partial-coloring
sampler to $B_T\operatorname{diag}(d)$, whose entries remain in $[0,1]$.
The choice between the two directions uses only $y$ and the sampled $z$.
The sampler returns at least $|T|/4$ nonzero coordinates, so each
update makes at least $|T|/8$ coordinates integral. We repeat until
every coordinate is integral; the geometric decrease bounds the
number of calls.

\begin{algorithm}[H]
\caption{\textsc{Round}: a common fractional-rounding routine}
\label[algorithm]{alg:fractional-rounding}
\small
\begin{algorithmic}[1]
\Require Initial vector $x\in[0,1]^q$ and a procedure $\mathcal S(T,d)$
  returning $z\in\{-1,0,1\}^T$ with at least $|T|/4$ nonzero coordinates.
\Ensure An integral vector $y\in\{0,1\}^q$.
\Function{Round}{$x,\mathcal S$}
  \State $y\gets x$.
  \While{$y$ has fractional coordinates}
    \State $T\gets\{g:0<y_g<1\}$.
    \State $d_g\gets\min\{y_g,1-y_g\}$ for $g\in T$.
    \State $z\gets\mathcal S(T,d)$.
    \State Replace $y$ by whichever of $y+d\odot z$ and $y-d\odot z$
      has more integral coordinates; favor $+$ in a tie.
  \EndWhile
  \State \Return $y$.
\EndFunction
\end{algorithmic}
\end{algorithm}

\begin{lemma}[Rounding from partial colorings]\label[lemma]{lem:rounding}
Let $B\in[0,1]^{n\times q}$ and $q\le m$.
Consider any realized execution of \cref{alg:fractional-rounding}
in which every sampler call returns a vector
$z$ with at least $|T|/4$ nonzero coordinates and
\[
 \norm{B_T(d\odot z)}_\infty\le3w+E_0,
\]
where $E_0\ge0$ is the fixed additional error per sampler call and
$w$ is given by
\eqref{eq:kernel-width} for $n$ rows and $|T|$ columns.
Then it makes at most $J=\lceil\log_{8/7}m\rceil+1$ sampler calls,
and its output satisfies
\[
 \norm{B(y-x)}_\infty
 \le C\sqrt n+JE_0.
\]
\end{lemma}
\begin{proof}
The choice of direction makes at least $|T|/8$ coordinates integral,
so the number of fractional coordinates decreases by a factor of
at most $7/8$ per call. Since $m(7/8)^J<1$, all coordinates
are integral after at most $J$ calls.

Let $w_j$ be the width at call $j$.
The number of fractional coordinates lies in any interval $[u,2u)$
for only a constant number of calls.
For sizes in $[2^rn,2^{r+1}n)$, their total width is at most
$C\sqrt n\,2^{-r/2}$. For sizes in
$[2^{-r-1}n,2^{-r}n)$, it is at most
$C\sqrt n\,2^{-r/4}$. Consequently,
\[
 \sum_j w_j
 \le C\sqrt n\sum_{r\ge0}\bigl(2^{-r/2}+2^{-r/4}\bigr)
 \le C\sqrt n.
\]
Each update has row error at most $3w_j+E_0$ in either direction.
Summing these errors over at most $J$ calls proves the claim.
\end{proof}

The same routine inherits conditional privacy and failure bounds
from its sampler. We state this for an arbitrary private input $X$:
the target matrix $B(X)$ used to measure error may itself depend on
$X$, as it will in \cref{sec:fairness}.

\begin{corollary}[Private rounding from a local sampler]
\label[corollary]{cor:private-rounding-oracle}
In the setting of \cref{lem:rounding}, let the initial vector $x$
and parameters $n,q,m,E_0$ be public, and let
$B=B(X)\in[0,1]^{n\times q}$ be fixed by the input $X$.
Let $\varepsilon_0\ge0$ and $0\le\beta_0\le1$, and recall that
$J=\lceil\log_{8/7}m\rceil+1$ bounds the number of sampler calls in
\textnormal{\textsc{Round}}.
Suppose all access to $X$ is through $\mathcal S$, which always
returns at least $|T|/4$ nonzero coordinates. Conditional on preceding
outputs, suppose each call is $\varepsilon_0$-DP with respect to
adjacency on $X$ and meets the error requirement of
\cref{lem:rounding} with probability at least $1-\beta_0$.
Then \textnormal{\textsc{Round}} is $J\varepsilon_0$-DP and satisfies that lemma's
error bound with probability at least $1-J\beta_0$.
Moreover, if the conditional sampler distributions agree after every
history possible on both inputs, then the rounding output has the same
distribution on both inputs.
\end{corollary}
\begin{proof}
All updates and stopping decisions depend only on the public inputs
and preceding sampler outputs. Adaptive composition over at most
$J$ calls gives privacy, and a union bound followed by
\cref{lem:rounding} gives accuracy. Identical conditional sampler
laws give identical transcript laws, proving the last assertion.
\end{proof}

We now supply the private partial-coloring sampler to this routine.

\begin{proposition}[Private fractional rounding]\label[proposition]{prop:round}
For every public $x\in[0,1]^q$, $q\le m$, there is an entry-level
$\varepsilon_0$-DP algorithm that, on input $B\in[0,1]^{n\times q}$, returns
$y\in\{0,1\}^q$ such that
\[
 \norm{B(y-x)}_\infty\le C\left(
 \sqrt n+\frac{\log^2 m}{\varepsilon_0}\log\frac{nm}{\beta_0}\right)
\]
with probability at least $1-\beta_0$, for $0<\varepsilon_0\le1$ and
$0<\beta_0<1/2$.
\end{proposition}

Set
\[
 J=\lceil\log_{8/7}m\rceil+1,\qquad
 \lambda=\frac{\varepsilon_0}{2J},\qquad
 E_0=\frac1\lambda\log\frac{33nJ(1+m\lambda)}{\beta_0}.
\]
Define $\textsc{PrivateRound}(B,x,\varepsilon_0,\beta_0;m)$ to be
$\textsc{Round}(x,\mathcal S)$ with
\[
 \mathcal S(T,d)=\textsc{PartialColor}
       (B_T\operatorname{diag}(d),\lambda,w),
\]
where $w$ is given by \eqref{eq:kernel-width} for $n$ rows and
$|T|$ columns.

\begin{proof}[Proof of \cref{prop:round}]
Condition on the preceding partial-coloring outputs. The current
$y,T,d,w$ are fixed on both neighboring inputs, so the progress test
also retains the same pairs. Scaling columns by $d_g\le1$ preserves
entry adjacency, and \cref{lem:gibbs} gives privacy cost
$2\lambda=\varepsilon_0/J$ per call.

Every returned coloring makes the required progress. For each row,
the tail bound in \cref{lem:gibbs} gives
\[
 \Prob\!\left[|(B_T(d\odot(\sigma-\tau)))_i|>3w+E_0\right]
 \le33(1+m\lambda)e^{-\lambda E_0}
 =\frac{\beta_0}{nJ}.
\]
A union bound over the rows gives error failure probability at most
$\beta_0/J$ per call. Thus \cref{cor:private-rounding-oracle} gives
$\varepsilon_0$-DP and, with probability at least $1-\beta_0$,
\[
 \norm{B(y-x)}_\infty
 \le C\sqrt n+JE_0
 \le C\sqrt n+\frac{C\log^2m}{\varepsilon_0}\log\frac{nm}{\beta_0}.
\]
Here $J=O(\log m)$, $\lambda\le1$, and
$\log(33nJ(1+m\lambda)/\beta_0)=O(\log(nm/\beta_0))$.
\end{proof}

Taking $x=(1/2,\ldots,1/2)$ and returning $2y-\one$ gives a private
signing with the same asymptotic discrepancy bound.

\subsection{From two colors to multicolor discrepancy}

We turn fractional rounding into a partition with $k$ bundles by
adapting the balanced recursion of \citet{DS03}. A node with $t$
intended bundles splits its goods in proportions
$\lfloor t/2\rfloor/t$ and $\lceil t/2\rceil/t$.
The following lemma isolates how rounding errors and privacy costs
accumulate in this tree.

\begin{lemma}[Recursive splitting]\label[lemma]{lem:tree-rounding}
Recursively split $([m],k)$, giving each node $(T,t)$, $t\ge2$,
two children with $\lfloor t/2\rfloor$ and $\lceil t/2\rceil$
intended bundles; empty nodes may be completed deterministically.
All path sums below run over internal nodes on root-to-leaf paths.

If each child $(W,s)$ of $(T,t)$ satisfies
\begin{equation}\label{eq:tree-split-error}
 \left|v_i(W)-\frac{s}{t}v_i(T)\right|\le E_t
 \qquad\text{for every row }i,
\end{equation}
where $E_t\ge0$, then every final bundle satisfies
\[
 \left|v_i(A_j)-\frac{v_i([m])}{k}\right|
 \le3\max_{\mathrm{path}}\sum_t\frac{E_t}{t}.
\]

For privacy, suppose each neighboring pair of private inputs admits a fixed
set $G$ of at most $\ell$ goods. Conditional on any preceding split
history possible on both inputs, hence fixing the same current
node $(T,t)$, assume the conditional distributions of the next split
agree when $T\cap G=\varnothing$ and otherwise assign probabilities
within a factor $e^{\varepsilon_t}$ to every output event.
Then the allocation is
\[
 \left(\ell\max_{\mathrm{path}}\sum_t\varepsilon_t\right)\text{-DP}.
\]
\end{lemma}
\begin{proof}
Track the value per intended bundle. Since $s\ge t/3$, dividing
\eqref{eq:tree-split-error} by $s$ gives
\[
 \left|\frac{v_i(W)}s-\frac{v_i(T)}t\right|\le\frac{3E_t}{t}.
\]
These differences telescope from the root to each final bundle:
\[
 v_i(A_j)-\frac{v_i([m])}{k}
 =\sum_{(T,t)\to(W,s)\text{ on the path}}
   \left(\frac{v_i(W)}s-\frac{v_i(T)}t\right).
\]
Taking absolute values proves the error bound. Empty subtrees
contribute zero.

For privacy, fix neighboring inputs and a common transcript of
split outputs. The nodes containing any one of the at most $\ell$
specified goods form a root-to-leaf path. All other nodes have
identical conditional laws. Multiplying conditional likelihood
ratios bounds the logarithm of the transcript's likelihood ratio by
\[
 \sum_{\substack{(T,t):\,T\cap G\ne\varnothing}}\varepsilon_t
 \le \sum_{g\in G}\sum_{\substack{(T,t):\,g\in T}}\varepsilon_t
 \le \ell\max_{\mathrm{path}}\sum_t\varepsilon_t.
\]
Here the sums on the left are over internal nodes in the fixed
transcript. Counting a node once for each affected good it contains
can only increase the bound; each inner sum then follows one path.
Exponentiating gives the claimed privacy guarantee.

Each child has at most $2/3$ as many intended bundles as its
parent. Read the internal nodes from the leaf toward
the root: their counts grow geometrically, starting at a value
of at least two. Summing the two geometric series gives
\begin{equation}\label{eq:path-sums}
 \sum_{\mathrm{path}}t^{-1/2}
 \le\frac{1}{\sqrt2(1-\sqrt{2/3})}<4,
 \qquad
 \sum_{\mathrm{path}}t^{-1}
 \le\frac{1/2}{1-2/3}=\frac32.\qedhere
\end{equation}
\end{proof}

We now give a common partition routine for this section and
\cref{sec:fairness}. Its input is a private fractional rounder.
The matrix whose rows we balance may depend on the private input;
privacy and locality are requirements on the rounder itself.

\begin{proposition}[Partitioning from private rounding]
\label[proposition]{prop:private-partition}
Let $X$ be a private input and $B=B(X)\in[0,1]^{N\times m}$.
Suppose that, for every public $T\subseteq[m]$, $x\in[0,1]^T$,
$0<\varepsilon_0\le1$, and $0<\beta_0<1/2$, a rounder
$\mathcal R_X(T,x,\varepsilon_0,\beta_0)$ is $\varepsilon_0$-DP on $X$ and returns
$y\in\{0,1\}^T$ with
\[
 \norm{B_T(X)(y-x)}_\infty
 \le D_0+\frac{D_1(\beta_0)}{\varepsilon_0}
\]
with probability at least $1-\beta_0$, where $D_0$ and $D_1(\beta_0)$
are public nonnegative bounds.
For each neighboring pair $X,X'$, suppose there is a fixed set
$G\subseteq[m]$, $|G|\le\ell$, with public $\ell\ge1$, such that
the rounder laws agree whenever $T\cap G=\varnothing$.
All guarantees hold conditional on preceding outputs. The locality
condition compares the rounder's output distributions after a history
possible on both inputs.

Then \cref{alg:multicolor} is $\varepsilon$-DP and, with probability
at least $1-\beta$, every final bundle satisfies
\[
 \norm{B(X)(\one_{A_j}-k^{-1}\one)}_\infty
 \le C\left(D_0+\frac{\ell}{\varepsilon}
                    D_1\!\left(\frac{\beta}{k-1}\right)\right).
\]
\end{proposition}

Splits near the root receive smaller privacy budgets because their
errors are divided among more final bundles. Write
$\textsc{PrivatePartition}(\mathcal R_X,k,\varepsilon,\beta,\ell;m)$
for the following routine.

\begin{algorithm}[H]
\caption{\textsc{PrivatePartition}: partition using a private rounder}
\label[algorithm]{alg:multicolor}
\small
\begin{algorithmic}[1]
\Require Rounder $\mathcal R_X$ as in \cref{prop:private-partition};
  public $m$, $2\le k<m$, $\ell\ge1$, $0<\varepsilon\le1$, $0<\beta<1/2$.
\Ensure An $\varepsilon$-DP partition with the guarantee of \cref{prop:private-partition}.
\State $\beta_0\gets\beta/(k-1)$ \Comment{Shared failure budget for each split}
\Function{Split}{$T,t$}
  \If{$t=1$}
    \State \Return $(T)$
  \EndIf
  \If{$T=\varnothing$}
    \State \Return a tuple of $t$ empty bundles
  \EndIf
  \State $x\gets(\lfloor t/2\rfloor/t)\one\in[0,1]^T$;
    $\varepsilon_t\gets\varepsilon/(4\ell\sqrt t)$
  \State $y\gets\mathcal R_X(T,x,\varepsilon_t,\beta_0)$, using fresh randomness.
  \State $W\gets\{g\in T:y_g=1\}$
  \State \Return the concatenation of \Call{Split}{$W,\lfloor t/2\rfloor$}
    and \Call{Split}{$T\setminus W,\lceil t/2\rceil$}.
\EndFunction
\State \Return $\Call{Split}{[m],k}$
\end{algorithmic}
\end{algorithm}

\begin{proof}[Proof of \cref{prop:private-partition}]
Condition on a common split history. The current set and starting
vector are fixed on both inputs. The rounder is $\varepsilon_t$-DP, and
its law agrees on the two inputs when the current set avoids $G$.
Thus \cref{lem:tree-rounding} bounds total privacy loss by
\[
 \ell\max_{\mathrm{path}}\sum_t\varepsilon_t
 =\frac{\varepsilon}{4}
       \max_{\mathrm{path}}\sum_t t^{-1/2}<\varepsilon.
\]
There are at most $k-1$ rounding calls. A conditional union bound
gives simultaneous success with probability at least $1-\beta$.
The two children have opposite rounding errors, so on this event
both satisfy \eqref{eq:tree-split-error} with
$E_t=D_0+D_1(\beta_0)/\varepsilon_t$. Apply \cref{lem:tree-rounding} to the
rows of $B(X)$. By \eqref{eq:path-sums}, the final error is at most
\[
 3D_0\max_{\mathrm{path}}\sum_t t^{-1}
 +\frac{12\ell D_1(\beta_0)}{\varepsilon}
       \max_{\mathrm{path}}\sum_t t^{-1/2}
 \le C\left(D_0+\frac{\ell D_1(\beta_0)}{\varepsilon}\right).
\]
\end{proof}

\begin{proof}[Proof of \cref{thm:main}]
Use $X=V$, $B(V)=V$, and
\[
 \mathcal R_V(T,x,\varepsilon_0,\beta_0)
 =\textsc{PrivateRound}(V_T,x,\varepsilon_0,\beta_0;m).
\]
By \cref{prop:round}, this rounder has
$D_0=C\sqrt n$ and
$D_1(\beta_0)=C\log^2m\log(nm/\beta_0)$.
It accesses only $V_T$, so its laws agree whenever $T$ excludes
the changed good; hence $\ell=1$.
\Cref{prop:private-partition} gives equal-share error
\[
 C\left(\sqrt n+\frac{\log^2m}{\varepsilon}
                         \log\frac{nkm}{\beta}\right).
\]
Every pair of bundle values differs by at most twice this quantity,
proving the discrepancy bound. The expected running time follows
from \cref{sec:sampling}.
\end{proof}

\section{Consensus fairness}\label{sec:fairness}
To turn discrepancy into a deletion guarantee, we balance both the
number of an agent's most valuable goods and the utility of her
remaining goods. The difficulty is that the normalization used for
the second quantity is sensitive to one utility change. We overcome
this by averaging the sampler's weights over cutoffs chosen by rank.

\begin{theorem}[Private consensus fairness]\label{thm:real-fairness}
For nonnegative additive utilities, an entry-level
$\varepsilon$-DP mechanism achieves consensus fairness up to $c$ goods
with probability at least $1-\beta$, where
\[
 c=O\!\left(
 \sqrt n
 +\frac{\log^2 m}{\varepsilon}\log\frac{nkm}{\beta}
 \right).
\]
The bound can be capped by $\lceil m/k\rceil$.
The mechanism has expected running time
$m^{O(n)}\operatorname{poly}(n,k)$ in the computation model of
\cref{sec:sampling}.
\end{theorem}

\subsection{The nonprivate discrepancy-to-fairness reduction}
\label{sec:nonprivate-fairness}
A small difference in utility need not mean that few goods can remove
it. For example, a bundle can contain arbitrarily many equal-valued
goods with total value one. It exceeds a zero-valued bundle by only
one unit, but every good must be deleted to remove that difference.
The reduction of
\citet[Theorem~3.2]{MS22} addresses this by giving every bundle enough
valuable goods to cover any deficit in its remaining utility.

\paragraph{Two bounded rows per agent.}
Fix an integer $q_0\in[m]$. For each agent $i$, let $H_i$ be the set of her
$q_0$ highest-valued goods, with ties broken by decreasing public
label. We call these her \emph{head goods}. Let $p_i^*$ be the
smallest value among them. Define a reference row by
\[
 v_i^*(g)=
 \begin{cases}
 v_i(g)/p_i^*,&g\notin H_i,\ p_i^*>0,\\
 0,&\text{otherwise}.
 \end{cases}
\]
Every good outside $H_i$ has value at most $p_i^*$, so this row lies
in $[0,1]^m$. The two rows associated with agent $i$ are therefore
$\one_{H_i}$, which counts head goods, and $v_i^*$, which measures
the remaining utility in units of $p_i^*$.

\begin{lemma}[From row balance to consensus fairness]\label[lemma]{lem:head-deletion}
Let $E\ge1$ and let $q_0\in[m]$ satisfy $q_0/k\ge3E$.
Suppose an allocation
satisfies
\[
 \left||H_i\cap A_j|-\frac{q_0}{k}\right|\le E,
 \qquad
 \left|v_i^*(A_j)-\frac{v_i^*([m])}{k}\right|\le E
\]
for every agent $i$ and bundle $A_j$.
Then
$v_i(A_j)\ge v_i(A_\ell\setminus H_i)$ for every $i,j,\ell$.
For each agent $i$ and bundle $A_\ell$, deleting $A_\ell\cap H_i$
from $A_\ell$ removes at most $q_0/k+E$ goods. The same deletion
set works for every bundle $A_j$, so the
allocation is consensus EF$\lceil q_0/k+E\rceil$.
\end{lemma}

\begin{proof}
For every agent $i$, the head-row bound gives
\[
 2E\le\frac{q_0}{k}-E
 \le |H_i\cap A_j|
 \le\frac{q_0}{k}+E.
\]
The reference-row bound gives
$|v_i^*(A_j)-v_i^*(A_\ell)|\le2E$.
Fix agent $i$ and an ordered pair $j,\ell$. When $p_i^*>0$,
\[
 \begin{aligned}
 v_i(A_j)
 &\ge p_i^*\bigl(|H_i\cap A_j|+v_i^*(A_j)\bigr)\\
 &\ge p_i^*\bigl(2E+v_i^*(A_j)\bigr)
 \ge p_i^*v_i^*(A_\ell)
 =v_i(A_\ell\setminus H_i).
 \end{aligned}
\]
If $p_i^*=0$, every good outside $H_i$ has value zero, so the same
deletion works directly. In either case, at most $q_0/k+E$ head goods
are deleted.
\end{proof}

In particular, the choice $q_0=\lceil3kE\rceil\le m$ gives
$q_0/k+E\le4E+1$ and hence consensus EF$\lceil4E+1\rceil$.

\paragraph{What changes under privacy.}
This reduction cannot be fed directly into \cref{thm:main}.
Changing one utility entry changes head membership at at most two
goods, but it can also change $p_i^*$. That change rescales every
nonzero coordinate of $v_i^*$, so the resulting matrices need not be
entry-adjacent.

We construct a private local sampler whose error is measured in
these same two rows. It uses a hidden cutoff among the active
non-head goods and accounts for the omitted goods in its error bound.
Its analysis uses the contraction stability and coloring guarantees
of \cref{lem:gibbs}, together with the mixture bound in
part~(ii) of \cref{lem:row-factor}.

\subsection{A local coloring from noisy ranks}\label{sec:hidden-threshold}
The auxiliary rows may omit a few non-head coefficients, provided
we account for their contribution to the rounding error. We use this freedom to choose a noisy rank among the
active \emph{non-head} goods. The proposed rank has the same public
distribution for every agent; acceptance depends on her utilities.

\paragraph{Auxiliary rows from rank cutoffs.}
Fix the active set $T$ and public coordinate weights
$d\in[0,1]^T$. In the rounding algorithm,
$d_g=\min\{y_g,1-y_g\}$ for the current fractional vector $y$,
fixed once preceding outputs are fixed.
All vector operations in this subsection are restricted to $T$.
For agent $i$, order only $T\setminus H_i$ by decreasing utility
and then decreasing public label. For a valid rank
$1\le s\le|T\setminus H_i|$, let $p_{i,s}$ be the value of
the $s$th good. Omit the first $s$ goods of this order and set
\[
 v_{i,s}(g)=
 \begin{cases}
 v_i(g)/p_{i,s},&
   \text{$g$ occurs after rank $s$ in this order},\ p_{i,s}>0,\\
 0,&\text{otherwise},
 \end{cases}
 \qquad g\in T.
\]
For all other integer ranks, set $v_{i,s}=0$. These rows lie in
$[0,1]^T$; a zero cutoff also gives the zero row.
The omitted goods remain in the allocation; only their coefficients
in this auxiliary row are set to zero. The next lemma bounds the
resulting error.

\begin{lemma}[Error from a rank cutoff]\label[lemma]{lem:nearby-cutoff}
For every integer rank $s\ge1$ and
$z\in\{-1,0,1\}^T$,
\[
 |(d\odot v_i^*)\cdot z|
 \le |(d\odot v_{i,s})\cdot z|+s.
\]
\end{lemma}

\begin{proof}
If $p_i^*=0$, the reference row is zero. Otherwise, for a valid
rank, $p_{i,s}\le p_i^*$ because the cutoff good is outside $H_i$.
Multiplying $v_{i,s}$ by $p_{i,s}/p_i^*\in[0,1]$ recovers $v_i^*$
except on the $s$ omitted non-head goods. Each omitted weighted
coefficient is at most one, which gives the stated bound.
If $s>|T\setminus H_i|$, there are fewer than $s$ active non-head
goods, and the zero row gives the same bound.
\end{proof}

\begin{lemma}[Stability of adjacent cutoffs]\label[lemma]{lem:cutoff-stability}
Fix public $q_0\in[m]$, $T\subseteq[m]$, and $d\in[0,1]^T$.
Let two utility profiles differ only in one entry of agent $i$,
and construct her head sets and cutoff rows separately on the two
profiles, using primes for the second. For every integer $s$, there
is a scalar $\rho\in[0,1]$ such that
\[
 \norm{d\odot v_{i,s-1}-\rho\,d\odot v'_{i,s}}_1\le4.
\]
\end{lemma}

\begin{proof}
The reason for comparing ranks $s-1$ and $s$ is to make the first
cutoff at least as large as the second. Dividing by the larger
cutoff then contracts the row, and a short count bounds the
coordinates that can still differ.

To make this precise, let $g$ be the changed good. The global head
set can change only at $g$ and one displaced good $g'$. The active
non-head lists agree after removing at most one exceptional candidate
from each. If head membership is unchanged, remove $g$ wherever it
occurs. If membership changes, remove the non-head candidate among
$g,g'$ from each list, when active. All remaining goods have
unchanged values and the same order, including the public rule for
ties. Call this the common list. In particular, the two original
list lengths differ by at most one.

First suppose that the primed rank $s\ge2$ is valid. Among its first
$s$ candidates, at least $s-1$ belong to the common list. Their
unprimed values are unchanged and are all at least $p'_{i,s}$.
Consequently, rank $s-1$ exists on the first input and
$p_{i,s-1}\ge p'_{i,s}$. If $p_{i,s-1}>0$, take
$\rho=p'_{i,s}/p_{i,s-1}\in[0,1]$. On any common good $h$ retained
below both cutoffs, when $p'_{i,s}>0$ we have
\[
 \rho\,v'_{i,s}(h)
 =\frac{p'_{i,s}}{p_{i,s-1}}\frac{v_i(h)}{p'_{i,s}}
 =v_{i,s-1}(h).
\]
If $p'_{i,s}=0$, such a good has value zero, so the same equality
holds directly with $\rho=0$. Common goods omitted by both cutoffs
also have equal, zero coefficients.

It remains to count the goods omitted on just one side. Removing
the exceptional candidate from the first $s-1$ unprimed goods leaves
a prefix of the common list of length $s-2$ or $s-1$. Doing the same
for the first $s$ primed goods leaves a prefix of length $s-1$ or
$s$. These prefixes differ by at most two goods. Together with the
at most two exceptional goods, this leaves at most four coordinates
where the unprimed row and the contracted primed row can differ.
Multiplication by $d_h\in[0,1]$ keeps every coefficient in $[0,1]$,
so each coordinate contributes at most one to the claimed norm.
If $p_{i,s-1}=0$, the cutoff inequality forces $p'_{i,s}=0$ as well;
both residual rows then vanish, and we can take $\rho=0$.

Finally, the invalid ranks require no cutoff comparison. For
$s\le1$, the unprimed row is zero, so take $\rho=0$. If $s$ exceeds
the length of the primed list, the primed row is zero. The unprimed
list has length at most $s$, so either rank $s-1$ is invalid or
omitting its first $s-1$ goods leaves at most one coordinate.
Taking $\rho=0$ gives the bound in this case too.
\end{proof}

\paragraph{The sampler.}
Fix a public integer upper rank $r\ge1$, rate $0<\lambda\le1$,
and width $w>0$. The privacy proof uses the adjacent-rank comparison
in \cref{lem:cutoff-stability}. Allowing nonpositive ranks keeps a
downward shift within the proposal support, and a geometric law
makes its probability ratio constant. We therefore propose a
geometric distance below $r$:
\[
 \pi(s)=(1-e^{-\lambda})e^{-\lambda(r-s)},\qquad
 s\in\mathbb Z,\ s\le r,
\]
and set $\pi(s)=0$ for $s>r$. The proposed rank always satisfies
$s\le r$, and summing the geometric series gives
\[
 \Prob[s\le0]=e^{-\lambda r}.
\]
When $s\ge1$, \cref{lem:nearby-cutoff} bounds the omitted-goods
error by $r$.
Using the row factors from \cref{lem:row-factor}, write
\[
 \begin{aligned}
 F_{i,s}(\sigma,\tau)
 &=F_{d\odot v_{i,s}}(\sigma,\tau),\\
 G(\sigma,\tau)
 &=\prod_i F_{d\odot\one_{H_i}}(\sigma,\tau).
 \end{aligned}
\]
Here $\sigma,\tau\in\{-1/2,1/2\}^T$.
The factor $G$ controls head counts, and $F_{i,s}$ controls
the residual at the proposed rank. Invalid ranks have $F_{i,s}=1$.
As in \cref{alg:partial-coloring}, we use the progress event
$\mathcal P$ with $q=|T|$.

\begin{algorithm}[H]
\caption{A local coloring from noisy ranks}
\label[algorithm]{alg:hidden-threshold}
\small
\begin{algorithmic}[1]
\Require Nonempty active set $T$, public weights $d\in[0,1]^T$,
  head sets $(H_i)$, utilities on $T$, integer upper rank $r\ge1$,
  rate $0<\lambda\le1$, and width $w>0$.
\Ensure A partial coloring in $\{-1,0,1\}^T$ with at least $|T|/4$ nonzero coordinates.
\State For each agent $i$, order $T\setminus H_i$ by decreasing utility,
  breaking ties by decreasing public label.
\Repeat
  \State Independently for each agent $i$, draw $s_i\sim\pi$.
  \State Draw $(\sigma,\tau)\sim Q$, uniformly on
    $\{-1/2,1/2\}^T\times\{-1/2,1/2\}^T$.
  \State Reject and restart the trial if $|\{g\in T:\sigma_g\ne\tau_g\}|<|T|/4$.
  \State Accept the pair with probability
    $G(\sigma,\tau)\prod_iF_{i,s_i}(\sigma,\tau)$.
\Until{the pair is accepted}
\State \Return $\sigma-\tau$; discard the ranks.
\end{algorithmic}
\end{algorithm}

Every rejection, including a failed progress check, redraws the ranks
and the pair. Without the progress check, the accepted joint law is
\begin{equation}\label{eq:hidden-joint}
 \Prob[(s_1,\ldots,s_n),\sigma,\tau]\ \propto\
 Q(\sigma,\tau)G(\sigma,\tau)
 \prod_i\pi(s_i)F_{i,s_i}(\sigma,\tau).
\end{equation}
Write $M_i=\sum_{s\in\mathbb Z}\pi(s)F_{i,s}$ for the averaged
residual factor. This is a countable probability mixture in the sense
of \cref{lem:row-factor}: for a fixed pair, $M_i$ is the expected row
weight when the proposed rank is drawn from $\pi$.
The corresponding pair marginal, still without the progress check, is
\begin{equation}\label{eq:hidden-batch}
 \frac{Q(\sigma,\tau)G(\sigma,\tau)\prod_i M_i(\sigma,\tau)}
        {\E_Q[G\prod_i M_i]}.
\end{equation}
The actual sampler conditions this marginal on $\mathcal P$.

\begin{lemma}[A local coloring for the reference rows]\label[lemma]{lem:latent-ranks}
Fix public parameters $\varnothing\ne T\subseteq[m]$,
$d\in[0,1]^T$, $r\in\mathbb Z_{\ge1}$, $0<\lambda\le1$, $w>0$,
and a head-set size $q_0\in[m]$; let each $H_i$ be the agent's
top-$q_0$ set, computed from her utilities with the fixed tie-breaking
rule. Then \cref{alg:hidden-threshold} is entry-level $14\lambda$-DP
and always colors at least $|T|/4$ coordinates. Its law depends only
on the utilities on $T$ and the sets $H_i\cap T$.
If $w$ is chosen as in \eqref{eq:kernel-width} for $2n$ rows and
$|T|$ columns, then its accepted pair $(\sigma,\tau)$ satisfies,
for every agent $i$,
\[
 \Prob\!\left[
 \max\{|(d\odot\one_{H_i})\cdot(\sigma-\tau)|,
        |(d\odot v_i^*)\cdot(\sigma-\tau)|\}>3w+2r
 \right]\le99(1+m\lambda)e^{-\lambda r}.
\]
\end{lemma}

\begin{proof}
All weights are positive and pairs in $\mathcal P$ exist, so the
sampler terminates almost surely. Its progress check guarantees at least
$|T|/4$ nonzero coordinates. We first analyze the unfiltered law
\eqref{eq:hidden-batch}; until conditioning below, probabilities
refer to this law.

\emph{Accuracy.}
Fix agent $i$. A large error in her reference row forces a large
error at every positive rank that can be proposed: the omitted goods
can account for at most $r$ of the error. More precisely, if
$|(d\odot v_i^*)\cdot(\sigma-\tau)|>3w+2r$, then \cref{lem:nearby-cutoff} gives
\[
 |(d\odot v_{i,s})\cdot(\sigma-\tau)|>3w+r
 \qquad\text{for every }1\le s\le r.
\]
The left side is the absolute difference between the two
untransformed row projections. The map $\phi_w$ moves each
projection by at most $w$, so the transformed projections still
differ by more than $w+r$. The exponential tail of $K_{\lambda,w}$
therefore gives $F_{i,s}\le e^{-\lambda r}$ at every such rank.
The only remaining ranks are nonpositive; they have total proposal
mass $e^{-\lambda r}$ and factor one. Averaging these two parts shows
that $M_i\le2e^{-\lambda r}$ on the reference-error event.

Since $M_i$ is a probability mixture of row factors, the event bound
in part~(ii) of \cref{lem:row-factor}, with
$1+|T|\lambda\le1+m\lambda$, gives
\[
 \Prob[|(d\odot v_i^*)\cdot(\sigma-\tau)|>3w+2r]
 \le2(1+m\lambda)e^{-\lambda r}.
\]
For the head row there is no rank averaging or omitted-goods error.
If its error exceeds $3w+2r$, its transformed projections differ by
more than $w+2r$, so its factor in $G$ is at most
$e^{-2\lambda r}$. Removing and restoring that single factor gives
head failure probability at most $(1+m\lambda)e^{-2\lambda r}$.
A union bound gives failure probability at most
$3(1+m\lambda)e^{-\lambda r}$ under the unfiltered law.

\emph{Conditioning on progress.}
With the prescribed width, condition on any rank tuple under
\eqref{eq:hidden-joint}. The pair then has the unfiltered law
\eqref{eq:gibbs} for the $2n$ bounded rows
$(d\odot\one_{H_i},d\odot v_{i,s_i})_{i=1}^n$.
The progress estimate in \cref{lem:gibbs} applies to every such
tuple, including invalid ranks, whose rows are zero. Averaging
over the ranks therefore gives
\[
 \Prob[\mathcal P]\ge1-e^{-|T|/32}\ge\frac1{33}.
\]
Conditioning on $\mathcal P$ increases any event probability by at
most $33$. Applying this once to the preceding accuracy bound gives
the claimed $99(1+m\lambda)e^{-\lambda r}$ bound.

\emph{Privacy.}
Fix neighboring profiles differing in one entry of agent $i$,
and fix $(\sigma,\tau)$. Write primes for the second profile.
The cutoff values can change substantially, so we compare the
unprimed rank $s-1$ with the primed rank $s$.
\Cref{lem:cutoff-stability} expresses the first row as a contraction
of the second, up to an $\ell_1$ perturbation of at most four.
Contraction can only increase a row's weight, and the perturbation
costs at most a factor $e^{4\lambda}$. Thus
\eqref{eq:contraction} gives, for every integer $s$,
\[
 F_{i,s-1}\ge e^{-4\lambda}F'_{i,s}.
\]
This is a comparison at adjacent ranks, whereas $M_i$ and $M'_i$
average over the same proposal law. The geometric probabilities
bridge this difference: $\pi(s-1)=e^{-\lambda}\pi(s)$ for $s\le r$.
Moreover, shifting down by one keeps every rank within the proposal
support. We can therefore drop the nonnegative top term
$\pi(r)F_{i,r}$ from $M_i$ and reindex the rest to obtain
\[
 M_i\ge\sum_{s\le r}\pi(s-1)F_{i,s-1}
     \ge e^{-5\lambda}\sum_{s\le r}\pi(s)F'_{i,s}
     =e^{-5\lambda}M'_i.
\]
The factor $e^{5\lambda}$ accounts for the row perturbation and
the one-step shift of the geometric law. Interchanging the inputs
and applying the same downward shift gives the reverse inequality, so
$M_i'/M_i\in[e^{-5\lambda},e^{5\lambda}]$.

At most two head memberships change, so the perturbation bound in
\cref{lem:gibbs} gives
$G'/G\in[e^{-2\lambda},e^{2\lambda}]$.
All other agents' factors agree. Thus the unnormalized pair weight
changes by at most $e^{7\lambda}$ in either direction. The progress
event $\mathcal P$ is the same on both inputs. Summing this pointwise
comparison over the retained pairs gives the same bound for the
ratio of their normalizing constants. Dividing by these normalizers
costs at most one further factor $e^{7\lambda}$, giving
$14\lambda$-DP for the actual sampler, as in
\cref{lem:exponential-mechanism}. Returning $\sigma-\tau$ is
postprocessing.

Finally, $\mathcal P$ is public and all factors in
\eqref{eq:hidden-batch} are constructed from utilities on $T$ and
the sets $H_i\cap T$. The global reference row $v_i^*$ is used only
to state and analyze the error, so its cutoff does not introduce
any additional dependence in the output law.
\end{proof}

\subsection{The allocation algorithm}\label{sec:real-algorithm}
We supply the hidden-rank sampler to \textsc{Round}, then use the
common partition routine of \cref{prop:private-partition}.
Only the fractional rounder and the deletion guarantee need to
be specified here.

Set
\[
 E=C\left(\sqrt n+\frac{\log^2m}{\varepsilon}
                    \log\frac{64nkm}{\beta}\right),\qquad
 q_0=\lceil3kE\rceil.
\]
Choose the constant in $E$ sufficiently large, as justified below.
If $\lceil m/k\rceil\le\lceil4E+1\rceil$, return a fixed partition
with bundle sizes at most $\lceil m/k\rceil$. Otherwise
$m/k>\lceil4E+1\rceil$, so $m>q_0=\lceil3kE\rceil$.
We can therefore fix each agent's global head set $H_i$ of size $q_0$.

Let $B(V)\in[0,1]^{2n\times m}$ have the head and reference rows
$(\one_{H_i},v_i^*)_{i=1}^n$ from \cref{sec:nonprivate-fairness}.
For public $U\subseteq[m]$, $x\in[0,1]^U$, $0<\varepsilon_0\le1$, and
$0<\beta_0<1/2$, define a rounder
$\mathcal R_V(U,x,\varepsilon_0,\beta_0)$ as follows. Set
\[
 J=\lceil\log_{8/7}m\rceil+1,\qquad
 \lambda=\frac{\varepsilon_0}{16J},\qquad
 r=\left\lceil\frac1\lambda
       \log\frac{99nJ(1+m\lambda)}{\beta_0}\right\rceil.
\]
As in \textsc{Round}, $J$ bounds the number of sampler calls.
For the current fractional coordinates $T\subseteq U$, let
$\mathcal S(T,d)$ run \cref{alg:hidden-threshold} with
upper rank $r$, rate $\lambda$, and the width
\eqref{eq:kernel-width} for $2n$ rows and $|T|$ columns.
The rounder returns $\textsc{Round}(x,\mathcal S)$.
Each local call orders only the active non-head goods, and its
error bound includes the contribution of the omitted goods.
The allocation mechanism returns
$\textsc{PrivatePartition}(\mathcal R_V,k,\varepsilon,\beta,2;m)$.

\begin{proof}[Proof of \cref{thm:real-fairness}]
The fixed-partition branch is $0$-DP and consensus
EF$\lceil m/k\rceil$, with $\lceil m/k\rceil\le\lceil4E+1\rceil=O(E)$.
For the sampling branch, we combine the local guarantee of
\cref{lem:latent-ranks} with \cref{cor:private-rounding-oracle}
to verify the hypotheses of \cref{prop:private-partition}.

\paragraph{Privacy and locality of the rounder.}
A utility change affects local utilities or head membership only
at the changed good and at most one displaced head good.
By \cref{lem:latent-ranks}, each local call is $14\lambda$-DP
with respect to the original profile, conditional on preceding
outputs, and its law agrees on neighboring profiles when the active
set avoids both goods. \Cref{cor:private-rounding-oracle} gives
privacy cost $14J\lambda\le\varepsilon_0$ and preserves this locality.
Thus the rounder has at most $\ell=2$ affected goods. Its outputs
contain no ranks or cutoff values.

\paragraph{Error of the rounder.}
For each agent and each call, \cref{lem:latent-ranks} and the choice
of $r$ bound the probability of error exceeding $3w+2r$ in either
her head or reference row by
\[
 99(1+m\lambda)e^{-\lambda r}\le\frac{\beta_0}{nJ}.
\]
These bounds hold conditional on preceding outputs. Every returned
coloring meets the progress requirement, so a union bound over the
$n$ agents and at most $J$ calls bounds the probability of any error
violation by $\beta_0$. Applying \cref{cor:private-rounding-oracle}
to $B_U(V)$ with additional error $E_0=2r$ gives
\[
 \norm{B_U(V)(y-x)}_\infty
 \le C\sqrt{2n}+2Jr
\]
with probability at least $1-\beta_0$.
Since $\lambda=\varepsilon_0/(16J)\le1$ and $J\le8m$,
\[
 r\le\frac{16J}{\varepsilon_0}
        \log\frac{99nJ(1+m\lambda)}{\beta_0}+1
 \le\frac{CJ}{\varepsilon_0}\log\frac{nm}{\beta_0}.
\]
Using $J=O(\log m)$ therefore gives
\[
 C\sqrt{2n}+2Jr
 \le C\sqrt{2n}
      +\frac{C\log^2m}{\varepsilon_0}\log\frac{nm}{\beta_0}.
\]
The two factors of $J$ come from the rank cutoff in each call and
the number of calls. These estimates are uniform over all head-set
sizes and the resulting reference rows.

\paragraph{From the common partition to fairness.}
\Cref{prop:private-partition} now applies with
$D_0=C\sqrt{2n}$, $D_1(\beta_0)=C\log^2m\log(nm/\beta_0)$, and $\ell=2$.
It gives $\varepsilon$-DP and, with probability at least $1-\beta$,
equal-share error for every head and reference row at most
\[
 \begin{aligned}
 C\left(D_0+\frac{2}{\varepsilon}
             D_1\!\left(\frac{\beta}{k-1}\right)\right)
 &\le C\left(\sqrt{2n}+\frac{2\log^2m}{\varepsilon}
                          \log\frac{nm(k-1)}{\beta}\right)\\
 &\le C\left(\sqrt n+\frac{\log^2m}{\varepsilon}
                          \log\frac{64nkm}{\beta}\right)
 \le E,
 \end{aligned}
\]
where the last inequality is ensured by choosing the constant in $E$
large enough.
Uniformity in the head sets justifies choosing the constant in $E$
before fixing $q_0=\lceil3kE\rceil$.
\Cref{lem:head-deletion} therefore gives consensus
EF$\lceil4E+1\rceil$. The fixed-partition branch gives the cap
$\lceil m/k\rceil$, and \cref{sec:sampling} gives the expected running time.
\end{proof}

\section{Envy-freeness over a public alphabet}\label{sec:alphabet}

For ordinary envy-freeness, we can improve an allocation by changing
who owns each bundle. Combining these ownership changes with private
counts gives a polynomial-time mechanism for a public finite utility
alphabet. The deletion bound depends on the number of allowed values, independently
of their magnitudes or spacing.

Throughout this section, $k=n$ and
\[
 \mathcal V=\{0=\xi_0<\xi_1<\cdots<\xi_{D-1}\}
 \subseteq\mathbb R_{\ge0}
\]
is public and fixed independently of the input. The private utility matrix
is $V\in\mathcal V^{n\times m}$; adjacent inputs differ by an arbitrary
replacement of one entry within $\mathcal V$.
An alphabet that does not contain zero can be enlarged by adding it;
this changes its size by only one and preserves the asymptotic bound.

\begin{theorem}[Entry-private envy-freeness over a public alphabet]
\label{thm:alphabet-ef}
Let $n,m,D\ge2$, $\varepsilon>0$, and $0<\beta<1$ be public.
There is a pure entry-level $\varepsilon$-DP mechanism that always returns
a complete allocation and, with probability at least $1-\beta$, achieves
EF$c$ with $c\le\lceil m/n\rceil$ and
\[
 c=O\!\left(1+\frac{\log m\,(\log D)^{3/2}}{\varepsilon}
                  \log\frac{nmD}{\beta}\right).
\]
It runs in time polynomial in $n,m,D$ in the computation model
of \cref{sec:sampling}.
\end{theorem}

The idea is to adapt the classical envy-cycle elimination algorithm of \citet{LMMS04} so that it works with approximate bundle values. Instead of using the true value of each bundle, the algorithm uses a lower estimate that never decreases over time. As long as the gap between this estimate and the true value can always be covered by deleting at most $c_0$ goods, the usual envy-cycle argument goes through with only one extra deletion: the last good inserted into each bundle, yielding EF$(c_0+1)$.

We use classical tools for differential privacy under continual
observation \citep{DNPR10,CSS11} to maintain these estimates
throughout envy-cycle elimination. After each insertion, we update the
receiving bundle's noisy counts: for each agent, these count the goods
she values at or above each public threshold. Subtracting a uniform error
bound gives valid lower estimates, and taking running maxima makes
them monotone.

\subsection{Envy cycles with monotone lower estimates}\label{sec:alphabet-estimates}

For an integer $c_0\ge0$, let
\[
 v_i^{-c_0}(S)=
 \min_{\substack{R\subseteq S\\|R|\le c_0}}v_i(S\setminus R)
\]
be the value left after deleting the $\min\{c_0,|S|\}$ most valuable goods.
Thus EF$c_0$ means $v_i(A_i)\ge v_i^{-c_0}(A_j)$ for all $i,j$.

Keep $n$ named bundles $P_1,\ldots,P_n$ and record ownership by a
permutation $\sigma$: agent $i$ owns $P_{\sigma(i)}$. Goods are added
to bundles but never move between them. Cycle rotations change only
$\sigma$, so the estimates can follow the same bundles throughout.

For every agent $i$ and bundle $b$, we need an estimate $w_{ib}$ satisfying
\begin{equation}\label{eq:alphabet-monotone-sandwich}
 v_i^{-c_0}(P_b)\le w_{ib}\le v_i(P_b).
\end{equation}
The estimates start at zero and never decrease. They are updated
after insertions and remain fixed between insertions.

Run the envy-cycle procedure of \citet{LMMS04} using these estimates.
Draw a strict envy edge $i\to i'$ when
\[
 w_{i,\sigma(i)}<w_{i,\sigma(i')}.
\]
Before adding a good, eliminate directed cycles by giving each participant
the bundle she points to. Then choose a source, an agent with no incoming
edge, and add the good to her bundle. Use fixed public rules to choose
cycles and sources.

Cycle elimination terminates: while it runs, the estimates of each named
bundle are fixed, and each participant strictly increases her estimate
of her own bundle.
An agent has only $n$ possible bundle estimates, so there are at most
$n(n-1)$ cycle participations between insertions.

\begin{lemma}[Last insertion]\label[lemma]{lem:alphabet-last-insertion}
If the estimates satisfy \eqref{eq:alphabet-monotone-sandwich} and the
update rules above, the procedure returns an EF$(c_0+1)$ allocation.
\end{lemma}
\begin{proof}
Each agent's estimate of her own bundle never decreases: a cycle rotation
strictly increases it, and estimate updates can only increase it.

Fix a nonempty final bundle $P_b$, and let $g$ be its last inserted good.
Just before that insertion, its owner was a source. Every agent's own
estimate was therefore at least her estimate of $P_b\setminus\{g\}$.
The latter estimate was at least $v_i^{-c_0}(P_b\setminus\{g\})$ by
\eqref{eq:alphabet-monotone-sandwich}. Since own estimates never decrease
and never exceed true values, the final allocation satisfies
\[
 v_i(A_i)\ge v_i^{-c_0}(P_b\setminus\{g\})
 \qquad\text{for every }i.
\]
Deleting $g$ and at most $c_0$ further goods removes envy of $P_b$.
Empty bundles require no deletions.
\end{proof}

The last-insertion argument, and the running-maximum update below,
use only monotonicity once estimates satisfying
\eqref{eq:alphabet-monotone-sandwich} are available. Additivity enters
in constructing those estimates from tail counts: it gives the
weighted-sum identity and makes deleting the
$\min\{c_0,|S|\}$ most valuable goods truncate every tail count
simultaneously.

\paragraph{From count error to deletions.}
For $1\le\ell<D$, define the tail count by
\[
 C_i(S,\ell)=|\{g\in S:v_i(g)\ge\xi_\ell\}|.
\]
Writing each utility as a sum of consecutive alphabet gaps gives
\begin{equation}\label{eq:alphabet-tail-identity}
 \begin{split}
 v_i(S)&=\sum_{\ell=1}^{D-1}(\xi_\ell-\xi_{\ell-1})C_i(S,\ell),\\
 v_i^{-c_0}(S)&=\sum_{\ell=1}^{D-1}
                    (\xi_\ell-\xi_{\ell-1})(C_i(S,\ell)-c_0)_+,
 \end{split}
\end{equation}
The second identity holds because deleting
the $\min\{c_0,|S|\}$ most valuable goods removes
$\min\{c_0,C_i(S,\ell)\}$ goods at or above every threshold simultaneously.

\begin{lemma}[Tail-count estimates]\label[lemma]{lem:alphabet-tails}
Suppose $E\ge0$ and
$|\widehat C_i(S,\ell)-C_i(S,\ell)|\le E$ for every threshold.
Then, for $c_0=\lceil2E\rceil$, the lower estimate
\[
 \underline v_i(S)=\sum_{\ell=1}^{D-1}
                    (\xi_\ell-\xi_{\ell-1})(\widehat C_i(S,\ell)-E)_+
\]
satisfies
\[
 v_i^{-c_0}(S)\le\underline v_i(S)\le v_i(S).
\]
\end{lemma}
\begin{proof}
The count error bound gives
\[
 (C_i(S,\ell)-c_0)_+
 \le(\widehat C_i(S,\ell)-E)_+
 \le C_i(S,\ell).
\]
Multiply by $\xi_\ell-\xi_{\ell-1}$ and sum, using
\eqref{eq:alphabet-tail-identity}.
\end{proof}

One set of at most $c_0$ deleted goods covers every threshold. Neither the number
of thresholds nor the sizes of their gaps multiply this deletion budget.

\paragraph{Making the estimates monotone.}
Suppose the tail counts of each receiving bundle are estimated within
$E$ after every insertion. When bundle $b$ receives a good, compute
its lower estimates and set
\[
 w_{ib}\gets\max\{w_{ib},\underline v_i(P_b)\}
 \qquad\text{for every }i,
\]
starting from $w_{ib}=0$ and leaving all other estimates unchanged.
The new estimate supplies the lower bound in
\eqref{eq:alphabet-monotone-sandwich}. Earlier estimates remain below
the current true value because goods only accumulate. Their maximum
therefore preserves the upper bound too. All the hypotheses of
\cref{lem:alphabet-last-insertion} hold with $c_0=\lceil2E\rceil$.

Thus any counter with uniform error $E$ supplies monotone lower
estimates and, through the envy-cycle procedure, an
EF$(1+\lceil2E\rceil)$ allocation. This deterministic implication
is the only accuracy property needed from the private counter.

\subsection{Adaptive private tail counting}

We now construct the counter that supplies this uniform error bound.
We specialize the multidimensional range construction of
\citet[Section~7.2]{CSS11} to a rectangular time--alphabet grid,
and prove privacy even when earlier counter outputs determine where each good is sent.

\begin{lemma}[Adaptive tail counting]\label[lemma]{lem:alphabet-ranges}
Let $n,m,D\ge2$, $\varepsilon>0$, and $0<\beta<1$.
Fix a public order $g_1,\ldots,g_m$. At time $t$, a fixed public rule
chooses a destination $b_t\in[n]$ from earlier counter outputs,
before reading the utilities of $g_t$. Write
$P_b(t)=\{g_s:s\le t,\ b_s=b\}$.

There is a polynomial-time mechanism \textnormal{\textsc{TailCounter}} that,
after each insertion, provides estimates $\widehat C_{ib\ell}(t)$
of $C_i(P_b(t),\ell)$ for all $i,b\in[n]$ and $1\le\ell<D$.
Its joint output transcript is pure $\varepsilon$-DP under replacement
of one entry of $V\in\mathcal V^{n\times m}$. With probability at least
$1-\beta$, every estimate at every time has error at most
\begin{equation}\label{eq:alphabet-count-accuracy}
 E=\frac{C}{\varepsilon}\left[
   (\log m\,\log D)^{3/2}\sqrt{\log\frac{nmD}{\beta}}
   +\log m\,\log D\,\log\frac{nmD}{\beta}
 \right],
\end{equation}
for a sufficiently large universal constant $C$ fixed in advance.
All noises can be sampled before processing the input. Each count
error is a fixed linear combination of these noises, with coefficients
depending only on $i,b,t,\ell$. Queries may reuse noises, so their
errors can be correlated. In particular, the simultaneous
accuracy event depends only on the noises and holds uniformly over
the adaptive choices of destinations.
\end{lemma}

\begin{proof}
We encode each tail count as a rectangle sum, and answer these sums
by combining noisy counts from a fixed family of smaller rectangles.
For every pair $(i,b)$, represent the routed input by an
$m\times D$ binary array, with time along the rows and alphabet
positions $0,\ldots,D-1$ along the columns. If $b_t=b$, row $t$
has a single $1$ in the column indexed by the alphabet position of
$v_i(g_t)$; otherwise, the row is zero. Thus the array records the
utilities that agent $i$ assigns to goods sent to bundle $b$.
The count $C_i(P_b(t),\ell)$ is exactly the sum over
$[1,t]\times[\ell,D-1]$: the time prefix selects goods already
inserted, and the alphabet suffix selects utilities at least $\xi_\ell$.

Pad both coordinate ranges to powers of two with public zeros.
In each coordinate, form the binary interval tree by repeatedly
bisecting the full range. Products of a time interval and an alphabet
interval are dyadic rectangles. Set
\[
 L=(1+\lceil\log_2m\rceil)(1+\lceil\log_2D\rceil)
   =O(\log m\,\log D).
\]
Each cell belongs to at most $L$ rectangles, one for each pair of
tree levels. This bounds how many noisy counts a changed cell can
affect. The same quantity also bounds how many noisy counts we need
to answer a query: a prefix or suffix has a disjoint decomposition
into at most one interval per tree level, and taking products of
the two decompositions gives at most $L$ rectangles.

In every array, release each rectangle's exact count plus independent
$\operatorname{Lap}(2L/\varepsilon)$ noise when its time interval ends;
at that point all of its input entries have arrived. Intervals ending
after $m$ are never released. Decompose the query
$[1,t]\times[\ell,D-1]$, extending its alphabet suffix through the
padded zeros, into at most $L$ disjoint dyadic rectangles.
Every time interval in this decomposition ends by $t$, so all of
these rectangle counts are available when the query is made.
Their noisy counts sum to $\widehat C_{ib\ell}(t)$.
Each rectangle receives noise only once: later queries reuse its
released noisy count. This construction uses
$O(n^2mD)$ rectangle counts and polynomial total work.

\emph{Privacy.}
Consider the full transcript of noisy rectangle counts from all
$n^2$ arrays, in a fixed public release order. All tail-count answers
are postprocessing of this transcript, so it suffices to prove privacy
for the rectangle counts. The point to check is that routing depends
on previous answers: changing an entry could change later destinations
in two actual executions.

To compare transcript densities, however, fix neighboring utility
profiles and evaluate both densities at the same full transcript.
At every time, the preceding counter outputs are then the same, so
the public routing rule chooses the same destination on both profiles.
With these destinations fixed, replacing one utility moves a single
$1$ from its old alphabet column to its new column, within the array
for the affected agent and destination. Thus at most two cells of
one array change; every other cell in every array agrees.

For a released rectangle $R$, its exact count is determined by the
input and the preceding transcript: earlier outputs determine the
destinations of all goods in its time interval. If $\Delta_R$ is the
difference between these counts on the two profiles, fresh independent
Laplace noise bounds the conditional density ratio by
$\exp(\varepsilon|\Delta_R|/(2L))$.
Each of the two changed cells belongs to at most $L$ rectangles,
so $\sum_R|\Delta_R|\le2L$ across the entire transcript.
In particular, the bound does not grow with the number of arrays.
Factor each full transcript density into its conditional densities
in release order. Multiplying the bounds on these conditional ratios gives
\[
 \exp\!\left(\frac{\varepsilon}{2L}\sum_R|\Delta_R|\right)
 \le e^\varepsilon.
\]
This holds at every common transcript. Integrating over any measurable
set of transcripts proves pure $\varepsilon$-DP for the full transcript
and hence for the counter's answers, including under adaptive destination
choices.

\emph{Accuracy.}
Let $\mathcal R(t,\ell)$ be the fixed dyadic decomposition of
$[1,t]\times[\ell,D-1]$, extended through the padded alphabet zeros,
and let $Z_{ib,R}$ be the noise on rectangle $R$ in array $(i,b)$.
The exact rectangle counts sum to the true tail count, so
\[
 \widehat C_{ib\ell}(t)-C_i(P_b(t),\ell)
 =\sum_{R\in\mathcal R(t,\ell)}Z_{ib,R}.
\]
This identity holds for every routing history. The rectangle family
depends only on $t,\ell$, so the event that all these indexed noise
sums have absolute value at most $E$ depends only on the noises.
Sampling the noises in advance therefore gives one accuracy event
that covers every adaptive choice of destinations.

Set $\beta_0=\beta/[n^2m(D-1)]$, the failure budget for each query.
Each fixed sum contains at most $L$ distinct, independent Laplace
noises of scale $2L/\varepsilon$. Thus the sum of squared scales is
at most $L(2L/\varepsilon)^2=4L^3/\varepsilon^2$, and the largest
scale is $2L/\varepsilon$. Substituting these two quantities into
the concentration bound of \citet[Corollary~2.9]{CSS11} gives,
with probability at least $1-\beta_0$, error at most
\[
 C\left[
 \sqrt{\frac{4L^3}{\varepsilon^2}\log\frac2{\beta_0}}
 +\frac{2L}{\varepsilon}\log\frac2{\beta_0}\right]
 =\frac{2C}{\varepsilon}\left[
 L^{3/2}\sqrt{\log\frac2{\beta_0}}
 +L\log\frac2{\beta_0}\right].
\]
Union-bound over all agents,
bundles, times, and suffixes. Different queries may share rectangle
noises and have correlated errors; the union bound requires no
independence between queries. Since
$\log(n^2mD/\beta)\le2\log(nmD/\beta)$, substituting
$L=O(\log m\,\log D)$ gives \eqref{eq:alphabet-count-accuracy}.
\end{proof}

\subsection{The allocation mechanism}

\Cref{alg:alphabet-allocation} combines the counter with envy-cycle
elimination and running maxima. The balanced-allocation branch gives
the cap $\lceil m/n\rceil$ by allowing deletion of an entire envied bundle.

\begin{algorithm}[H]
\caption{Private envy-cycle allocation over a public alphabet}
\label[algorithm]{alg:alphabet-allocation}
\small
\begin{algorithmic}[1]
\Require Valuation matrix $V\in\mathcal V^{n\times m}$; public alphabet $\mathcal V$,
  order $g_1,\ldots,g_m$, privacy $\varepsilon$, and failure probability $\beta$.
\Ensure An $\varepsilon$-DP allocation with the guarantee of \cref{thm:alphabet-ef}.
\State Set $E$ by \eqref{eq:alphabet-count-accuracy}.
\If{$\lceil m/n\rceil\le1+\lceil2E\rceil$}
  \State \Return a fixed allocation with bundle sizes at most $\lceil m/n\rceil$.
\EndIf
\State Initialize \textsc{TailCounter} with parameters $\varepsilon,\beta$.
\State $P_1,\ldots,P_n\gets\varnothing$; $\sigma(i)\gets i$ for all $i$;
  $w_{ib}\gets0$ for all $i,b$.
\For{$t=1,\ldots,m$}
  \While{the graph $i\to i'$ iff $w_{i,\sigma(i)}<w_{i,\sigma(i')}$ has a directed cycle}
    \State Choose a cycle $i_1\to\cdots\to i_r\to i_1$ by the public rule.
    \State Simultaneously set $\sigma(i_j)\gets\sigma(i_{j+1})$ for all $j$, with $i_{r+1}=i_1$.
  \EndWhile
  \State Choose a source by the public rule; let $b_t$ be the index of its bundle.
    \label{line:alphabet-destination}
  \State $P_{b_t}\gets P_{b_t}\cup\{g_t\}$.
  \State Read $(v_i(g_t))_{i=1}^n$ and insert it with destination $b_t$ into \textsc{TailCounter}.
  \For{$i=1,\ldots,n$}
    \State Obtain $\widehat C_{i b_t\ell}(t)$ for all $1\le\ell<D$ from \textsc{TailCounter}.
    \State $w_{i b_t}\gets\max\bigl\{w_{i b_t},\,
      \sum_{\ell=1}^{D-1}(\xi_\ell-\xi_{\ell-1})(\widehat C_{i b_t\ell}(t)-E)_+\bigr\}$.
      \label{line:alphabet-update}
  \EndFor
  \State Leave all other estimates unchanged.
\EndFor
\State \Return $A_i=P_{\sigma(i)}$ for each agent $i$.
\end{algorithmic}
\end{algorithm}

The counter's outputs are internal; only the final allocation is published.
Updating only the receiving bundle's estimates in
line~\ref{line:alphabet-update} requires $nm(D-1)$ tail-count queries in total.

\begin{proof}[Proof of \cref{thm:alphabet-ef}]
The fixed-allocation branch is $0$-DP and meets the stated cap.
We analyze the other branch.

\emph{Privacy.}
Every destination is chosen before reading the next utility column,
using only previous counter outputs. Ownership changes and the final
allocation also depend only on this transcript. \Cref{lem:alphabet-ranges}
therefore gives $\varepsilon$-DP by postprocessing.

\emph{Fairness.}
With probability at least $1-\beta$, \cref{lem:alphabet-ranges} bounds
all count errors by $E$. The deterministic reduction in
\cref{sec:alphabet-estimates} therefore gives EF$(1+\lceil2E\rceil)$.

Combining the two branches gives
\[
 c=\min\!\left\{\lceil m/n\rceil,\,1+\lceil2E\rceil\right\}.
\]
To simplify \eqref{eq:alphabet-count-accuracy}, use
$\log m\le\log(nmD/\beta)$ for the first term and
$\log D\ge\log2$ for the second:
\[
 \begin{aligned}
 (\log m\,\log D)^{3/2}\sqrt{\log\frac{nmD}{\beta}}
 &\le\log m\,(\log D)^{3/2}\log\frac{nmD}{\beta},\\
 \log m\,\log D\,\log\frac{nmD}{\beta}
 &\le\frac{\log m\,(\log D)^{3/2}}{\sqrt{\log2}}
       \log\frac{nmD}{\beta}.
 \end{aligned}
\]
Both terms in the error bound are therefore at most a constant times
\[
 \frac{\log m\,(\log D)^{3/2}}{\varepsilon}\log\frac{nmD}{\beta}.
\]
Retaining the additive constant gives the stated bound for every
$\varepsilon>0$.

\emph{Computation.}
The counter runs in polynomial time. There are $nm(D-1)$ tail-count
queries and at most $n(n-1)$ cycle participations between consecutive
insertions. Finding a cycle or source, updating the estimates, and
all remaining operations take polynomial time in $n,m,D$ in the
model of \cref{sec:sampling}.
\end{proof}

\section{Lower bound}\label{sec:lower-bounds}

The following theorem rules out EF$c$ guarantees with constant $c$
with success probability $0.99$, answering a question of
\citet[Section~5]{MS25}.

\begin{theorem}[Entry-private envy-freeness lower bound]\label{thm:ef-lower}
Let $n\ge2$, $0<\varepsilon\le1$, $0<\beta<1/50$, and
$m\ge n\log(n/\beta)/\varepsilon^2$.
If an entry-level $\varepsilon$-DP mechanism returns an EF$c$ allocation
with probability at least $1-\beta$ on every nonnegative additive profile, then
\[
 c=\Omega\!\left(\frac{\log(n/\beta)}{\varepsilon}\right).
\]
The hard profiles have binary utilities, and allocations may be arbitrary.
\end{theorem}

The proof will compare the mechanism's envy gaps with those from
independent binary values: an agent values each own good with
probability $1-p$ and each foreign good with probability $p$, where
$p=(1+e^\varepsilon)^{-1}$. Between two bundles of $r$ goods, the
resulting gap has law $\operatorname{Bin}(2r,p)-r$. The following
lemma bounds the probability of a large positive gap despite this
bias in favor of the own bundle.

\begin{lemma}[A binomial envy tail]\label[lemma]{lem:binomial-envy-tail}
Let $0<\varepsilon\le1$, let $r\ge1$ be an integer, and put
$p=(1+e^\varepsilon)^{-1}$. Then
\[
 \Prob[\operatorname{Bin}(2r,p)-r>\varepsilon r/4]
 \ge\frac1{40}e^{-3\varepsilon^2r}.
\]
\end{lemma}
\begin{proof}
We first increase the Bernoulli parameter
so that the desired event has constant probability, then use
Cauchy--Schwarz to bound the cost of changing back. Set
\[
 p'=\frac12+\frac{\varepsilon}{8}+\frac1{4\sqrt r}\le\frac78.
\]
If $X\sim\operatorname{Bin}(2r,p')$, then
$\mathbb E X=r+\varepsilon r/4+\sqrt r/2$ and
$\operatorname{Var}X\le r/2$.
Cantelli's one-sided variance inequality therefore gives
\[
 \Prob[X-r\le\varepsilon r/4]
 =\Prob[X-\mathbb E X\le-\sqrt r/2]
 \le\frac{\operatorname{Var}X}{\operatorname{Var}X+r/4}
 \le\frac{r/2}{r/2+r/4}=\frac23.
\]
Taking the complement yields $\Prob[X-r>\varepsilon r/4]\ge1/3$.

Write $\mathbb E_p$ and $\Prob_p$ for expectation and probability
when $X$ has law $\operatorname{Bin}(2r,p)$.
Let $R$ be the likelihood ratio of the binomial law with parameter
$p'$ to the one with parameter $p$. Independence of the $2r$
Bernoulli draws gives
\[
 \mathbb E_p R^2
 =\left(\frac{(p')^2}{p}+\frac{(1-p')^2}{1-p}\right)^{2r}
 =\left(1+\frac{(p'-p)^2}{p(1-p)}\right)^{2r}.
\]
Since $\varepsilon\le1$, we have
$p\in[1/(1+e),1/2]\subseteq[1/4,1/2]$, so $p(1-p)\ge3/16$.
Also,
\[
 \frac12-p=\int_0^\varepsilon\frac{e^u}{(1+e^u)^2}\,du
 \le\int_0^\varepsilon\frac14\,du=\frac{\varepsilon}{4},
\]
where $(1+e^u)^2\ge4e^u$ bounds the integrand. Consequently,
$p'-p\le3\varepsilon/8+1/(4\sqrt r)$.
Using $1+x\le e^x$ and $(a+b)^2\le2a^2+2b^2$, we obtain
\[
 \log\mathbb E_p R^2
 \le\frac{32r}{3}
       \left(\frac{3\varepsilon}{8}+\frac1{4\sqrt r}\right)^2
 \le3\varepsilon^2r+\frac43.
\]

The event probability under parameter $p'$ is
$\mathbb E_p[R\mathbf1_{\{X-r>\varepsilon r/4\}}]$.
Cauchy--Schwarz bounds its square by
$\mathbb E_p R^2\,\Prob_p[X-r>\varepsilon r/4]$. Hence
\[
 \Prob[\operatorname{Bin}(2r,p)-r>\varepsilon r/4]
 \ge\frac19\exp\!\left(-3\varepsilon^2r-\frac43\right)
 \ge\frac1{40}e^{-3\varepsilon^2r},
\]
where the last step uses $9e^{4/3}<40$.
\end{proof}

\begin{proof}[Proof of \cref{thm:ef-lower}]
For binary utilities, an envy gap greater than $c$ violates EF$c$,
since deleting at most $c$ goods removes at most $c$ units of value.
We will draw a random binary profile and show that, when $c$ is
too small, the probability of such a gap exceeds $\beta$, forcing
some fixed profile to violate the mechanism's guarantee.

Conditioning on the allocation can make the utility entries
dependent. Privacy will let us compare this conditional law with
independent utilities biased in favor of each agent's own bundle.
Even in this experiment, many agents have independent chances
of a large envy gap.

Fix an integer $q\ge1$, to be chosen in Step~3, and set $m=nq$.
The argument extends directly to any $m\ge nq$ by padding with
goods valued at zero by every agent and removing them from the
returned bundles. This induces a mechanism on $nq$ goods with
the same privacy and EF$c$ guarantees: appending zero entries
preserves entry-level adjacency, and discarding zero-valued goods
preserves bundle values and deletion guarantees.

Let $\mathcal M$ be an entry-level $\varepsilon$-DP mechanism on
these $m=nq$ goods with EF$c$ success probability at least $1-\beta$
on every profile. Draw each entry $v_i(g)$, $(i,g)\in[n]\times[m]$,
independently and uniformly from $\{0,1\}$. Write $V$ for this
random profile and $\mathcal A=\mathcal M(V)$ for the random output.
This experiment includes both the randomness of $V$ and the
mechanism's internal randomness.
Our target is an envy gap greater than $\varepsilon q/4$.

\emph{Step 1: use privacy to obtain an independent comparison.}
Fix an allocation $A=(A_1,\ldots,A_n)$ with
$\Prob[\mathcal A=A]>0$, and write $\Prob_A$ for probabilities
conditional on $\mathcal A=A$. Thus the bundles are fixed while
the profile $V$ remains random.
Revealing utility entries below imposes additional conditions on
this same random profile.

Fix one entry $X=v_i(g)$. Partition all other utility entries
$v_j(h)$, with $(j,h)\ne(i,g)$, into two vectors: $Y$ contains
the revealed entries, and $Z$ contains the unrevealed entries.
This partition ranges across all agents and excludes $X$.
The values $Y=y$ specify the utilities at the revealed positions;
$Z=z$ completes the other entries, still leaving $X$ unspecified.

Fixing both $Y=y$ and $Z=z$ specifies every entry except $X$.
Let $V^{(0)}$ and $V^{(1)}$ be the profiles obtained by setting
$X$ to zero and one. Before observing the output, these profiles
are equally likely. They differ in one entry, so privacy gives
\[
 e^{-\varepsilon}
 \le\frac{\Prob[\mathcal M(V^{(1)})=A]}
          {\Prob[\mathcal M(V^{(0)})=A]}
 \le e^\varepsilon.
\]
We consider conditioning events of positive probability; privacy
then makes both likelihoods positive. By Bayes' rule, the ratio
equals the conditional odds that $X=1$. Hence
\[
 p\le\Prob_A[X=1\mid Y=y,Z=z]\le1-p,
 \qquad p=\frac1{1+e^\varepsilon}.
\]

When only $Y=y$ is revealed, the law of total probability gives
\[
 \Prob_A[X=1\mid Y=y]
 =\sum_z \Prob_A[Z=z\mid Y=y]\,
          \Prob_A[X=1\mid Y=y,Z=z].
\]
The sum includes only completions of positive conditional
probability. Each conditional probability
$\Prob_A[X=1\mid Y=y,Z=z]$ lies in $[p,1-p]$, and the weights
$\Prob_A[Z=z\mid Y=y]$ sum to one. Thus the probability after a
partial revelation is a weighted average of probabilities in
$[p,1-p]$, and also lies in that interval.

To make every envy gap nondecreasing in each bit, define
\[
 T_{ig}=
 \begin{cases}
  1-v_i(g),&g\in A_i,\\
  v_i(g),&g\in [m]\setminus A_i.
 \end{cases}
\]
On an agent's own bundle, we swap zero and one: a transformed
one means that her own good is worthless. On a foreign bundle,
a one still means that the good is valuable to her. For $i\ne j$,
\[
 \begin{aligned}
 v_i(A_j)-v_i(A_i)
 &=\sum_{g\in A_j}T_{ig}-\sum_{g\in A_i}(1-T_{ig})\\
 &=\sum_{g\in A_j}T_{ig}+\sum_{g\in A_i}T_{ig}-|A_i|.
 \end{aligned}
\]
Since $|A_i|$ is fixed, changing a transformed bit from zero to
one increases this gap by one if it occurs in either sum, and
leaves it unchanged otherwise. Hence decreasing any transformed
bit can only decrease or leave unchanged every envy gap.

Each transformed bit has conditional probability at least $p$
of being one after other transformed bits are revealed. For an
unchanged bit this uses the original lower bound; for a
complemented bit it uses the upper bound:
\[
 \Prob_A[1-v_i(g)=1\mid\text{revealed entries}]
 =1-\Prob_A[v_i(g)=1\mid\text{revealed entries}]
 \ge p.
\]
Because $A$ specifies which entries were complemented, revealing
a transformed value determines its original value, so the
conditional bounds apply to these revelations too.

Order the agent--good pairs by a fixed rule and write
$T_1,\ldots,T_{nm}$ for the transformed bits in this order.
Generate their conditional law one bit at a time. After generating
$t_1,\ldots,t_{\ell-1}$, let
\[
 \theta_\ell=
 \Prob_A[T_\ell=1\mid T_1=t_1,\ldots,T_{\ell-1}=t_{\ell-1}]
 \ge p.
\]
Draw independent uniform variables $U_1,\ldots,U_{nm}$ in $[0,1]$,
and at each position set
\[
 T_\ell=\one_{\{U_\ell\le\theta_\ell\}},
 \qquad
 \widetilde T_\ell=\one_{\{U_\ell\le p\}}.
\]
The first rule uses the conditional probability given the values
already generated, so it reproduces the conditional law,
including its dependencies. The second uses only the independent
draws and the fixed threshold $p$, giving independent
Bernoulli$(p)$ bits. Whenever $\widetilde T_\ell=1$, we have
$U_\ell\le p\le\theta_\ell$, so $T_\ell=1$ as well. Thus
$\widetilde T_\ell\le T_\ell$ at every position.

Return to agent--good indices and undo the complements to define
the comparison utilities:
\[
 \widetilde v_i(g)=
 \begin{cases}
  1-\widetilde T_{ig},&g\in A_i,\\
  \widetilde T_{ig},&g\in[m]\setminus A_i.
 \end{cases}
\]
These entries are independent, with own values Bernoulli$(1-p)$
and foreign values Bernoulli$(p)$. We evaluate them on the fixed
allocation $A$, without running the mechanism again. The envy
identity and $\widetilde T_{ig}\le T_{ig}$ give, for every $i\ne j$,
\[
 \widetilde v_i(A_j)-\widetilde v_i(A_i)
 \le v_i(A_j)-v_i(A_i).
\]
Thus a large gap under $\widetilde v$ implies a large gap in the
coupled sample from the original conditional law.

\emph{Step 2: find many independent opportunities for envy.}
Keep $A$ fixed and order the agents by increasing $|A_i|$, with
fixed tie-breaking. Let $I$ be the first $\lceil n/2\rceil$ agents,
and let $i_\star$ be the last. Since $n\ge2$, $i_\star\notin I$.
We compare each $i\in I$ with the largest bundle $A_{i_\star}$.

The largest bundle has at least the average $q$ goods.
Each selected bundle has at most $2q$ goods: otherwise it and
all $\lfloor n/2\rfloor$ unselected bundles would each have more
than $2q$ goods, exceeding the total $nq$. For $i\in I$, set
$r_i=\max\{q,|A_i|\}$. Then
\[
 q\le r_i\le2q,\qquad |A_{i_\star}|\ge r_i.
\]
To apply the binomial lemma, reduce each comparison to two lists
of $r_i$ values. Choose $F_i\subseteq A_{i_\star}$ of size $r_i$
by a fixed rule depending only on $A$. On the own side, append
$r_i-|A_i|$ auxiliary Bernoulli$(1-p)$ values
$D_{i,1},\ldots,D_{i,r_i-|A_i|}$. These are independent of all
comparison utilities and of each other, including across agents.
If $r_i=|A_i|$, no auxiliary values are added.
The padding is only for this calculation; the allocation stays
$A$. Define the reduced gap
\[
 \begin{aligned}
 R_i
 &=\sum_{g\in F_i}\widetilde v_i(g)
   -\left(\sum_{g\in A_i}\widetilde v_i(g)
     +\sum_{t=1}^{r_i-|A_i|}D_{i,t}\right)\\
 &\le\widetilde v_i(A_{i_\star})-\widetilde v_i(A_i).
 \end{aligned}
\]
Discarding foreign values and adding own values can only reduce
envy, giving the inequality.

The retained foreign values are $r_i$ independent Bernoulli$(p)$
bits. The padded own values are $r_i$ independent
Bernoulli$(1-p)$ bits, so their zero indicators are independent
Bernoulli$(p)$ bits. The two lists use disjoint utility entries
and independent auxiliary values. The padded own value equals
$r_i$ minus its number of zeros. Therefore
\[
 R_i=(\text{number of foreign ones})
       +(\text{number of padded own zeros})-r_i
\]
has distribution $\operatorname{Bin}(2r_i,p)-r_i$.
The reduced gaps are also independent across agents: each uses
a different utility row and independent padding. Sharing a
foreign good causes no dependence, since $\widetilde v_i(g)$
and $\widetilde v_j(g)$ are independent for $i\ne j$.

In this comparison experiment with $A$ fixed,
\cref{lem:binomial-envy-tail} and $q\le r_i\le2q$ give
\[
 \begin{aligned}
 \Prob[R_i>\varepsilon q/4]
 &\ge\Prob[\operatorname{Bin}(2r_i,p)-r_i>\varepsilon r_i/4]\\
 &\ge\frac1{40}e^{-3\varepsilon^2r_i}
 \ge\frac1{40}e^{-6\varepsilon^2q}.
 \end{aligned}
\]
By independence, $1-x\le e^{-x}$, and $|I|\ge n/2$,
\[
 \begin{aligned}
 \Prob[R_i\le\varepsilon q/4\text{ for every }i\in I]
 &=\prod_{i\in I}\bigl(1-\Prob[R_i>\varepsilon q/4]\bigr)\\
 &\le\exp\!\left(-\sum_{i\in I}\Prob[R_i>\varepsilon q/4]\right)\\
 &\le\exp\!\left(-\frac n{80}e^{-6\varepsilon^2q}\right).
 \end{aligned}
\]
Taking the complement and using the two gap comparisons---from
$R_i$ to $\widetilde v$, then from $\widetilde v$ to the conditional
profile $V$---gives, for every released allocation $A$,
\[
 \Prob_A\!\left[\exists i\ne j:
       v_i(A_j)-v_i(A_i)>\varepsilon q/4\right]
 \ge1-\exp\!\left(-\frac n{80}e^{-6\varepsilon^2q}\right).
\]

\emph{Step 3: choose the scale and conclude.}
To obtain a gap of order $\log(n/\beta)/\varepsilon$ while keeping
the probability of at least one large gap above $\beta$, choose
\[
 q=\left\lfloor\frac{\log(n/\beta)}{1000\varepsilon^2}\right\rfloor.
\]
The theorem ensures that at least $nq$ goods are available.
Suppose first that $q\ge1$, so the construction and padding
argument apply. Since $\varepsilon^2q\le\log(n/\beta)/1000$
and $6/1000<1/64$,
\[
 \frac n{80}e^{-6\varepsilon^2q}
 \ge\frac n{80}(n/\beta)^{-1/64}
 =\frac\beta{80}(n/\beta)^{63/64}
 >\frac98\beta.
\]
The last inequality uses $n/\beta>100$ and $100^{63/64}>90$.
Thus the conditional probability of a gap greater than
$\varepsilon q/4$ exceeds $\beta$, because
\[
 1-e^{-9\beta/8}
 \ge\frac{9\beta}{8+9\beta}>\beta,
\]
using $e^x\ge1+x$ and $\beta<1/50<1/9$.

If $c\le\varepsilon q/4$, this gap violates EF$c$. Averaging over
the outputs $A$ gives joint failure probability greater than
$\beta$. This is also the average of the mechanism's failure
probabilities on the fixed binary profiles, so some profile has
failure probability greater than $\beta$, a contradiction.
Hence $c>\varepsilon q/4$, and $\lfloor x\rfloor\ge x/2$ for
$x\ge1$ gives
\[
 c>\frac{\varepsilon q}{4}
 \ge\frac{\log(n/\beta)}{8000\varepsilon}.
\]
If $q=0$, then
$\log(n/\beta)/\varepsilon<1000\varepsilon\le1000$.
On the original instance, give one good value one for every
agent and all other goods value zero. Every allocation leaves
an agent envying the owner of that good until it is deleted.
Thus $c\ge1\ge\log(n/\beta)/(8000\varepsilon)$, completing the proof.
\end{proof}

In the regime of \cref{tab:results}, the assumption $\beta\ge n^{-a}$
gives $\log(n/\beta)\le(a+1)\log n$. Thus, when
$n\ge(a+1)/\varepsilon^2$,
\[
 m\ge n^2\log n
 \ge\frac{n(a+1)\log n}{\varepsilon^2}
 \ge\frac{n\log(n/\beta)}{\varepsilon^2},
\]
so the theorem applies and gives the claimed $\Omega(\log n)$ bound.
For the remaining $2\le n<(a+1)/\varepsilon^2$, the quantity
$\log n$ is bounded by a constant depending only on the fixed
$a$ and $\varepsilon$. One good valued by every agent forces $c\ge1$,
which gives the same $\Omega(\log n)$ bound for these values of $n$.

\section{Row-level privacy}\label{sec:row-privacy}
An allocation that ignores utilities is automatically row-private.
We choose between a balanced partition and independent uniform
assignment, using only public parameters.

\begin{theorem}[Row-private consensus fairness]\label{thm:row-fairness}
Let $n\ge1$, $m\ge2$, $k\ge2$, and $0<\beta<1/2$ be public.
For nonnegative additive utilities, a row-level $0$-DP
mechanism achieves consensus EF$c$ with probability
at least $1-\beta$, where $c\le\lceil m/k\rceil$ and
\[
 c=O\!\left(\min\left\{
 \left\lceil\frac mk\right\rceil,
 1+\sqrt{\frac mk\log\frac{nk}{\beta}}
 \right\}\right).
\]
The mechanism uses $O(m+k)$ time in the computation model of
\cref{sec:sampling} and does not inspect any utility values.
\end{theorem}

Set
\[
 T=\left\lceil2\sqrt{\frac{m}{k}\log\frac{4nk}{\beta}}\right\rceil.
\]
The mechanism returns a fixed partition with bundle sizes at most
$\lceil m/k\rceil$ if $\lceil m/k\rceil\le4T$; otherwise it assigns
each good independently and uniformly to one of the $k$ bundles.

\begin{proof}
Both branches ignore utilities, so the mechanism is row-level $0$-DP,
and each takes $O(m+k)$ time. The fixed partition is consensus
EF$\lceil m/k\rceil$ by deleting an entire envied bundle.

In the random branch, $\lceil m/k\rceil>4T$. Since $4T$ is an
integer, this implies $m/k>4T$. This branch condition lets us
use a threshold with only a square-root term in Bernstein's bound.
For any fixed row $b\in[0,1]^m$ and bundle $j$, the sum
$b(A_j)=\sum_g b_g\one_{\{g\in A_j\}}$ has mean $b([m])/k$
and variance at most $\sum_g b_g^2/k\le m/k$.
Its summands are independent and lie in $[0,1]$. Since $T<m/k$,
Bernstein's inequality and the definition of $T$ give
\[
 \Prob\!\left[\left|b(A_j)-\frac{b([m])}{k}\right|>T\right]
 \le2\exp\!\left(-\frac{T^2}{2m/k+2T/3}\right)
 \le2\exp\!\left(-\frac{kT^2}{4m}\right)
 \le\frac{\beta}{2nk}.
\]

For each agent, take her top $3kT$ goods as head goods; this is
possible because $m>4kT$. The resulting $2n$ head and reference
rows $(\one_{H_i},v_i^*)_{i=1}^n$ are defined as in
\cref{sec:nonprivate-fairness} and used only in the analysis. A union
bound over these rows and the $k$ bundles shows that, with probability
at least $1-\beta$, each bundle's sum in each row lies within $T$ of
that row's average per bundle. Apply \cref{lem:head-deletion} with
$E=T$: each deletion set has size at most
$3kT/k+T=4T$, proving consensus EF$(4T)$.

Thus the mechanism achieves
\[
 c=\min\{\lceil m/k\rceil,4T\}
 \le\min\left\{\left\lceil\frac{m}{k}\right\rceil,
                4+8\sqrt{\frac{m}{k}\log\frac{4nk}{\beta}}\right\}.
\]
Since $\log(4nk/\beta)=\Theta(\log(nk/\beta))$, this proves the
stated bound, including the exact cap $c\le\lceil m/k\rceil$.
\end{proof}

\begin{corollary}[Optimal order under row-level privacy]\label[corollary]{cor:row-optimal}
Fix a constant $a>0$. For $m\ge n\ge2$, $0<\varepsilon\le1$, and
\[
 n^{-a}\le\beta<\frac{e^{-\varepsilon}}{200},
\]
the optimal worst-case deletion bound under row-level
$\varepsilon$-DP, with success probability at least $1-\beta$, is
\[
 \Theta\!\left(\min\left\{\frac mn,\sqrt{\frac mn\log n}\right\}\right).
\]
This holds for nonnegative additive utilities, both for ordinary
EF$c$ and for consensus EF$c$ with $n$ bundles.
The upper bound is achieved with privacy parameter zero;
its constant may depend on $a$.
\end{corollary}

\begin{proof}
Apply \cref{thm:row-fairness} with $k=n$.
The assumption $\beta\ge n^{-a}$ gives
$\log(n^2/\beta)=O_a(\log n)$.
Since $m/n\ge1$ and $n\ge2$, the ceiling and additive constant are
absorbed, giving the stated upper bound.

For the matching lower bound, \citet[Theorem~3.1]{MS25} show that a
row-level $\varepsilon$-DP mechanism that is EF$c$ with probability
strictly above $1-e^{-\varepsilon}/200$
on every binary instance must have
\[
 c=\Omega\!\left(\min\left\{\frac mn,\sqrt{\frac mn\log n}\right\}\right).
\]
Our assumption on $\beta$ gives this probability requirement.
Consensus EF$c$ implies ordinary EF$c$, so the same lower bound
applies to both notions.
\end{proof}

\appendix

\section{Sampling and running time}\label[appendix]{sec:sampling}
We count exact arithmetic, comparisons, floor operations, logarithms,
square roots, exponentials, and uniform or Laplace draws as unit-cost operations.
The polynomial running time of the alphabet mechanism is proved in
\cref{sec:alphabet}. Here we bound the cost of the rejection samplers
in \cref{sec:model,sec:fairness}.

For a sampler call on $q$ goods, recall that $Q$ is the uniform law
on pairs $(\sigma,\tau)\in\{-1/2,1/2\}^q\times\{-1/2,1/2\}^q$.
\Cref{alg:partial-coloring,alg:hidden-threshold} accept a proposal
with probability equal to its weight. Write $W(\sigma,\tau)$ for the
acceptance probability of a proposed pair, averaged over the proposed
ranks in the fairness sampler, and set $W=0$ when the progress
check fails. Each independent trial succeeds with
probability $\E_QW$. Summing over the possible numbers of preceding
rejections gives the probability of returning a fixed pair:
\[
 \sum_{j\ge1}(1-\E_QW)^{j-1}Q(\sigma,\tau)W(\sigma,\tau)
 =\frac{Q(\sigma,\tau)W(\sigma,\tau)}{\E_QW}.
\]
Thus the accepted pair has the normalized weighted proposal law.
Similarly, the probability of needing more than $j$ trials is
$(1-\E_QW)^j$, so
\[
 \E[\text{number of trials}]
 =\sum_{j\ge0}(1-\E_QW)^j=\frac1{\E_QW}.
\]

\paragraph{Evaluating a proposal.}
Checking progress costs $O(q)$ operations.
The explicit kernel in \cref{sec:partial-colorings} takes a constant
number of operations to evaluate. A trial forms the row sums and
evaluates one factor per row, using $O(nm)$ operations.
For the fairness sampler, a trial evaluates the head factors and
the residual factors at the proposed ranks; their averages $M_i$
over the cutoff ranks are used only in the analysis.

\paragraph{Acceptance probability.}
First omit the progress check. Starting from the uniform pair law,
add the rows one at a time.
By part~(i) of \cref{lem:row-factor}, each addition multiplies the normalizer
by at least $(1+q\lambda)^{-1}$. Thus for
$B\in[0,1]^{n\times q}$ and $\lambda,w>0$,
\begin{equation}\label{eq:normalizer}
 Z_B(\lambda,w)\ge(1+q\lambda)^{-n}.
\end{equation}
Part~(ii) gives the same cost for a probability mixture of row
factors. For the fairness sampler, applying this bound to the $n$ head
factors and the $n$ averaged residual factors $M_i$ in
\eqref{eq:hidden-batch} gives acceptance probability at least
$(1+q\lambda)^{-2n}$, including invalid cutoff ranks.
With the prescribed widths, the conditioning arguments in
\cref{lem:gibbs,lem:latent-ranks} show that the progress check retains
at least $1/33$ of this accepted mass. It therefore increases the
expected number of trials by at most a factor $33$.

\paragraph{Total running time.}
Every call in \cref{sec:model,sec:fairness} has
$\lambda\le1$ and $q\le m$. The preceding acceptance bounds give
at most $33(1+m)^n$ expected trials for discrepancy and
$33(1+m)^{2n}$ for consensus fairness.
The $O(k)$ recursive splits each use $O(\log m)$ partial-coloring calls.
Multiplying these call counts by the $O(nm)$ cost per trial gives
expected time $O(nmk\log m\,(1+m)^n)$ for discrepancy.

For consensus fairness, sort each agent's non-head goods once at
the start, then filter that order by the current active set in
$O(nm)$ operations per call. This gives exactly the ordering used
by the local sampler. Recall that $r$ is the public upper rank and
$\pi(s)=(1-e^{-\lambda})e^{-\lambda(r-s)}$ for integers $s\le r$.
To draw a rank from this law exactly, take a uniform
$R\in(0,1)$ and set
\[
 s=r-\left\lfloor-\frac{\log R}{\lambda}\right\rfloor.
\]
For every integer $j\ge0$,
\[
 \Prob\!\left[\left\lfloor-\frac{\log R}{\lambda}\right\rfloor\ge j\right]
 =\Prob[R\le e^{-\lambda j}]=e^{-\lambda j}.
\]
Subtracting successive tails gives
\[
 \Prob[s=r-j]
 =e^{-\lambda j}-e^{-\lambda(j+1)}
 =(1-e^{-\lambda})e^{-\lambda j}
 =\pi(r-j).
\]
Hence $s$ has the required geometric rank law.
Drawing the ranks and evaluating a proposal costs $O(nm)$ operations.
There are $O(k\log m)$ calls, so consensus fairness takes expected time
$O(nmk\log m\,(1+m)^{2n})=m^{O(n)}\operatorname{poly}(n,k)$.

\section*{Disclosure of AI use}
Most proofs in this paper were obtained using ChatGPT 5.6 Sol and
ChatGPT 6 Astra. We subsequently verified the proofs for correctness
and refined their exposition and arguments, also with the aid of
ChatGPT.

\end{document}